\documentclass[11pt]{article}
\usepackage{etex}
\usepackage[all]{xy}           
\usepackage{amsmath,amsfonts,amssymb,amsthm,textcomp}
\usepackage{xypic}
\usepackage{eucal}
\usepackage{epsfig}
\usepackage{tikz-cd}
\usepackage[colorlinks=true]{hyperref}
\usepackage{epic}
\usepackage{fancybox}
\usepackage{wrapfig}
\usepackage[font=footnotesize,labelsep=period]{caption}
\usepackage{float}
\usepackage{booktabs}
\usepackage[all]{xy}
\usepackage{fancyhdr}
\usepackage{appendix}
\usepackage{color}

\def\v #1{\vert #1\vert}             

\def\m #1 #2{(-1)^{{\v #1} {\v #2}}} 

\DeclareMathOperator{\Ima}{Im}

\theoremstyle{plain}
\newtheorem{theorem}{Theorem}

\newtheorem{lemma}[theorem]{Lemma}
\newtheorem{example}{Example}

\theoremstyle{definition}

\def\<#1>{\langle#1\rangle}

\begin{document}

\centerline{\Large \bf A Hamilton--Jacobi theory}
\vskip 0.2cm
\centerline{\Large \bf for implicit differential systems}
\vskip 0.5cm
\vskip 0.5cm

\centerline{O\u{g}ul Esen$^{\dagger}$, Manuel de Le\'on$^{\ddagger}$, Cristina
Sard\'on$^{*}$}
\vskip 0.5cm

\centerline{Department of Mathematics$^{\dagger}$}
\centerline{Gebze Technical University}
\centerline{41400 Gebze, Kocaeli, Turkey.}
\vskip 0.5cm

\centerline{Consejo Superior de Investigaciones Cient\'ificas$^{\ddagger}$}
\centerline{C/ Nicol\'as Cabrera, 13--15, 28049, Madrid. SPAIN}
\vskip 0.5cm

\centerline{Instituto de Ciencias Matem\'aticas, Campus Cantoblanco$^{*}$}
\centerline{Consejo Superior de Investigaciones Cient\'ificas}

\begin{abstract}
  In this paper, we propose a geometric Hamilton-Jacobi theory for systems of implicit differential equations.
 
 In particular, we are interested in implicit Hamiltonian systems, described in terms of Lagrangian submanifolds of $TT^*Q$ generated by Morse families.
 The implicit character implies the nonexistence of a Hamiltonian function
 describing the dynamics. This fact is here amended
 by a generating family of Morse functions which plays the role of a Hamiltonian. A Hamilton--Jacobi equation is obtained with the aid of this generating family of functions.
 To conclude, we apply our results to singular Lagrangians by employing the construction of special symplectic structures.

\end{abstract}

\section{Introduction}

Implicit differential equations (IDE) do not only arise in purely mathematical frameworks, as in relation with
minimizers of integrals in the calculus of variations, or as
intermediate steps for the integration of differential algebraic equations \cite{Hoef}, they do appear in many various areas in science as well. Their
applications are important in relativity, control theory, chemistry, etc.
For example, they describe exchanges of matter, energy, or information that vary in space and time.

Unfortunately, differential equations and, particularly, IDE cannot always be solved analytically. For this matter, different mathematical methods
have been precisely designed. It is desirable that these differential equations are reducible to quadratures, and many attempts have been tried
through Fuchsian, Lie's theory, etc. Nowadays, the methods of numerically integration has been increasingly developed.

The Hamilton--Jacobi theory (HJ theory) has been long known as a powerful problem solving tool \cite{Ar89,AbMa78}. It is particulary useful for identifying
conserved quantities for a mechanical system, which may be possible even when the mechanical problem itself cannot be solved completely.
Therefore, it constitutes an alternative way of finding solutions of Hamilton's equations. It is equivalent
to other classical formulations of mechanics and it roots in variational calculus. The action functions are solutions
of Hamilton--Jacobi equation (HJE). It is important to remark that the classical HJ theory only deals with explicit Hamiltonian systems, but, in the literature there exist tons of physical models governed by IDE. Hence, the necessity of constructing a Hamilton--Jacobi theory for implicit systems.

For example, recall the number of theories described by singular Lagrangians in the sense of Dirac-Bergmann \cite{AnderBerg,Be56,Dirac2}, including
systems appearing in gauge theories \cite{leonrodrigues}. The  Euler--Lagrange equations (EL) give rise to differential equations that are implicit, and because of the degeneracy of the Lagrangian they cannot be put in a normal form.
Some authors have introduced a geometric formalism for dealing with dynamical systems in their implicit form \cite{tulc2,tulc1}
and a unified approach for the Lagrangian description of (time-independent) constrained mechanical systems
is provided through a technique that generates IDE on $T^{*}Q$ from one-forms
defined on the total space of any fiber bundle over $TQ$ \cite{BaGrasMende}. Other authors have designed algorithms following the Driac-Bergmann prescription, to be able to deal with singular Hamiltonian and Lagrangian theories, see for example, the geometric Gotay--Nester algorithm \cite{gotay2,gotay5,gotay3,gotay0} (see brief description in our Appendix).
In the local coordinate formalism, the classical EL displaying conservative
or nonconservative force fields or subject to linear or nonlinear nonholonomic constraints, also arise in implicit form from
d'Alambert's principle of virtual work. In the geometric formalism, the corresponding equations should then be expected to arise
in implicit form equivalently from a suitable expression of the above principle.


Our main aim is to generalize the geometric Hamilton Jacobi explained for explicit systems to the realm of implicit systems on $T^*Q$.
We interpret IDE in terms of arbitrary submanifolds of a higher-order tangent bundle, particularly, Lagrangian submanifolds of $TT^*Q$ in the case of implicit Hamiltonian systems (IHS), not necessarily in the horizontal form. As an application, we shall concentrate the problem of Hamilton Jacobi theory for singular Lagrangian theories.  

Let us summarize the problem in more technical terms. We consider a first-order IDE as a submanifold $E$ of $TT^{*}Q$. We project $E$ to $TQ$ by the tangent mapping $T{\pi_Q}$ to a submanifold
$T{\pi_Q}(E)$ of $TQ$, which is another IDE on $Q$. The philosophy of the geometric Hamilton--Jacobi theory is to retrieve solutions
of $E$, provided the solutions of $T{\pi_Q}(E)$. In similar fashion as the
classical Hamilton-Jacobi theorem, in order to lift the solutions in $Q$
to $T^*Q$, we are still in need of a closed one-form $\gamma$ on $Q$, but two ingredients of the theory are missing. One is that the base manifold $C=\tau_{T^*Q}(E)$ is not necessarily the whole $T^*Q$, but possibly a proper submanifold of it. The second is the nonexistence of a Hamiltonian vector field due to the implicit character of the equations. In the classical theory, the major role of the Hamiltonian vector field is to connect the image space of $\gamma$ and the submanifold $E$. To overcome these two difficulties, we need to introduce a auxiliary section $\sigma$ of the fibration $\tau_{T^*Q}$ defined on $C\cap \Ima{\gamma}$ and taking values in $E$ .
If, particularly, the dynamics $E$ is a Lagrangian submanifold then, according to generalized Poincar\'{e} theorem \cite{BeTu80,Ja00,LiMa,TuUr99,We77}, there exists a Morse family (a family of generating functions) defined on the total space of a smooth bundle linked to $TT^*Q$ by means of a special symplectic structure. A Morse family also establishes a link from the base space $T^*Q$ to $E$, so that for this particular case, there is no need for an auxiliary section.

The plan of the manuscript is the following: in section 2 we review the fundamentals of Hamiltonian mechanics, Section 3
develops a geometric interpretation of dynamics as Lagrangian submanifolds and their generationg with the aid of Morse families of functions.
Section 4 illustrates the geometric HJ theory both for IDE and IHS. In section 5, we contemplate the construction of complete solutions. Section 6 concerns applications of our constructed theory to the case of degenerate Lagrangians.

We assume that functional analytic issues related with the present discussion are satisfied in order to highlight the main aspects of our theory. Accordingly, we assume that all manifolds are connected, all mathematical objects are real, smooth and globally defined.

\subsection{Notation chart}

Let $Q$ be the configuration space, $TQ$ is the tangent bundle, and $T^{*}Q$ is the cotangent  bundle. Consider the tangent and cotangent bundles of $TQ$ and $T^{*}Q$, these are the possibilities: $TTQ, T^{*}TQ, TT^{*}Q$ and $T^{*}T^{*}Q$.
 Here we can establish the canonical projections for the first order tangent and cotangent bundles, denoted as $\pi_Q:T^{*}Q\rightarrow Q$ and ${\tau_Q}:TQ\rightarrow Q$.
 Furthermore, consider the projections, $\pi_{T^{*}Q}:T^{*}T^{*}Q\rightarrow T^{*}Q$, $\tau_{T^{*}Q}:TT^{*}Q\rightarrow T^{*}Q$,
 $T{\pi_Q}:TT^{*}Q\rightarrow TQ$ and the two last projections $\pi_{TQ}:T^{*}TQ\rightarrow TQ$ and $\tau_{TQ}:TTQ\rightarrow TQ$. For the last case there is another
possibility $T{\tau_Q}:TTQ\rightarrow TQ$, and both possibilities are related through a diffeomorphism we shall devise in the following lines.

{\begin{center}
\begin{table}[H]{\footnotesize
  \noindent
\caption{{\small {\bf Canonical coordinates and symplectic forms on second-order tangent and cotangent spaces.} Consider $Q$ a mechanical
configuration manifold and note that we are assuming summation over repeated indices.}}
\label{table2}
\medskip
\noindent\hfill
\resizebox{\textwidth}{!}{\begin{minipage}{\textwidth}
\centering
\begin{tabular}{ l l l }

\hline
 &   \\[-1.5ex]
  Space& Coordinates & Symplectic forms \\[+1.0ex]
\hline
 &   \\[-1.5ex]
$Q$ &   $q^i$ & \\[2.0ex]
$TQ$  & $(q^i,\dot{q}^i)$  & \\[2.0ex]
$T^{*}Q$ & $(q^i,p_i)$ &  $\omega_Q= dq^i\wedge dp_i$\\[2.0ex]
$TT^{*}Q$ & $(q^i,p_i,\dot{q}^i,\dot{p}_i)$ &  $\omega_Q^{T}= d\dot{q}^i\wedge dp_i+dq^i\wedge d\dot{p}_i$\\[2.0ex]
$TTQ$ & $(q^i,\dot{q}^i,\delta q^i, \delta \dot{q}^i)$ & \\[2.0ex]
$T^{*}T^{*}Q$ & $(q^i,p_i,\alpha_i,\beta^i)$ & $\omega_{T^{*}Q}= dq^i\wedge d\alpha_i+dp_i\wedge d\beta^i$\\[2.0ex]
$T^{*}TQ$ & $(q^i,\dot{q}^i,a_i,b_i)$ & $\omega_{TQ}=dq^i\wedge da_i+d\dot{q}_i\wedge db_i$ \\[2.0ex]
\hline\\


\end{tabular}
  \end{minipage}}
\hfill}
\end{table}
\end{center}}

Let us recall the definition of the pullback bundle, as we refer to it in forthcoming sections. Let $(P,\pi,M)$ be a fiber bundle
and assume the existence of a differential mapping $\varphi$ from a manifold $P$ to the base manifold $M$. Define the following product manifold
$$ \varphi^*P=\{(n,p)\in N\times P : \varphi(n)=\pi(p)\}$$
and the surjective submersion $\varphi^*\pi$ which simply projects an element in $\varphi^*P$ to its first factor in $N$.
The triple $(\varphi^*P,\varphi^*\pi,N)$ is called the pullback bundle of $(P,\pi,M)$ via the mapping $\varphi$.
This definition can be summarized within the following commutative diagram.
\begin{equation}\label{pbb}
\xymatrix{
\varphi^*P \ar[rr] ^{\varepsilon} \ar[dd]^{\varphi^*\pi}&& P \ar[dd]^{\pi}\\ \\
N \ar [rr]_{\varphi}&& M
}
\end{equation}
Here, $\varepsilon$ is the projection which maps an element in $\varphi^*P$ to its second factor $P$.
Let us point out two particular cases which are important for the present work. 
If $\varphi$ is a diffeomorphism,  $\varphi^*P$  and $P$ become diffeomorphic.
If $M$ is an embedded submanifold of $P$, and $N$ is an embedded submanifold of $M$ in (\ref{pbb}), then it is evident that $\varphi^*P$ is an embedded submanifold of $P$. 
In this case, the map $\varepsilon$ plays the role of an embedding.

Henceforth, we refer to a general, arbitrary manifolds by $M$ or $N$, we will denote fiber bundles by $\pi:P\rightarrow N$, where $P$ is the complete space and $N$ is its projection by $\pi$ (also $M$ instead of $N$ indistinctly).
IDE will be denoted by $E$ generally, and any general Lagrangian submanifold is denoted by $S$. We will also refer by $E$ to IHS that are Lagrangian submanifolds generated by a Morse families. By $F$ we denote Morse families. In general, $P$ is the total space of a fiber bundle related with a Morse family $F$ and $R$ is the total space of a special symplectic structure.

\section{Fundamentals}

This section is dedicated for reviewing fundamentals of Hamiltonian mechanics from a geometric viewpoint, and basics of the Tulczyjew's triple. Here, we are setting
the notation we shall be using along the paper.

\subsection{Geometry of the cotangent bundle}

Consider a manifold $Q$ and a cotangent bundle $T^{*}Q$ with canonical projection $\pi_Q:T^{*}Q\rightarrow Q$. We denote by $(q^i,p_i)$
the fibered coordinates in $T^{*}Q$ such that $\pi_Q(q^i,p_i)=(q^i)$. A cotangent bundle is equipped with a canonical one-form $\theta_Q$ defined as follows:
\begin{equation}
 \langle X_{\alpha_q}, \theta_Q(\alpha_q)\rangle=\langle T\pi_Q(X_{\alpha_q}),\alpha_q\rangle,
\end{equation}
where $X_{\alpha_q}\in T_{\alpha_q}(T^{*}Q)$ and $\alpha_q\in T^{*}_qQ$.
In fibered coordinates, the canonical one-form reads $\theta_Q=p_idq^i$, which is known as the Liouville one-form on $T^{*}Q$.
Consider now the two-form $\omega_Q=-d\theta_Q$, namely, $\omega_Q=dq^i\wedge dp_i$. This two-form has the two properties
\begin{enumerate}
 \item It has maximal rank $2n$, where $n$ is the dimension of $Q$.
 \item $d\omega_Q=0$
\end{enumerate}
This two-form is called a symplectic two-form.
More generally, a symplectic manifold is a pair $(M,\omega)$ such that $\omega$ has maximal rank and $d\omega=0$.
Therefore, $(T^{*}Q,\omega_Q)$ is a symplectic manifold and $\omega_Q$ is the canonical symplectic form on $T^{*}Q$.

Given two symplectic manifolds $(M_1,\omega_1)$ and $(M_2,\omega_2)$ and a map $F:M_1\rightarrow M_2$, we say that $F$
is a symplectomorphism if $F^{*}\omega_2=\omega_1$.

\subsection{Hamiltonian dynamics on the cotangent bundle}

A Hamiltonian system on $T^*Q$ is determined by the triple $\left( T^{\ast }Q,\omega _{Q},H\right) $, where $H$ being the Hamiltonian function.
Geometrically, Hamilton's equations are defined by
\begin{equation}\label{geomHeq}
 \iota_{X_H}\omega_Q=dH,
\end{equation}
where $X_H$ is the Hamiltonian vector field associated with the Hamiltonian function $H$, and $\iota_{X_H}$ is the inner contraction operator. 

In Darboux's coordinates $(q^i,p_i)$ on $T^*Q$, with $i=1,\dots,n$ for an $n$-dimensional configuration manifold $Q$. The canonical one-form reads $\theta_Q=p_idq^i$ and the symplectic two-form turns out to be $\omega_Q=dq^i\wedge dp_i$. In this local picture, the Hamiltonian vector field $X_H$ is written as
\begin{equation} \label{HamEq}
 X_H=\frac{\partial H}{\partial p_i}\frac{\partial}{\partial q^i}-\frac{\partial H}{\partial q^i}\frac{\partial}{\partial p_i}
\end{equation}
whereas the Hamilton's equations (\ref{geomHeq}) turn out to be
\begin{equation}\label{hamileq22}
{\dot q}^i=\frac{\partial H}{\partial p_i},\qquad {\dot p}_i=-\frac{\partial H}{\partial q^i}.
 \end{equation}

Note that, the Hamilton's equations are explicit by definition. Hence, an IDE cannot be recast
in the classical Hamiltonian formalism presented in (\ref{hamileq22}). To deal with IHS, we shall redefine Hamiltonian systems in a more abstract framework, as we shall present in the forthcoming sections.

\subsection{Geometry of the tangent bundle}

Consider the manifold $Q$ and its tangent bundle $TQ$ together with its tangent bundle canonical projection $\tau_Q:TQ\rightarrow Q$.
We consider fibered coordinates $(q^i,\dot{q}^i)$ such that $\tau_Q(q^i,\dot{q}^i)=(q^i)$.
Given a function $f:Q\rightarrow \mathbb{R}$, we define its complete lift $f^{T}$ to $TQ$ as the function:
\begin{equation}
 f^{T}(v_q)=df(q)(v_q)\in \mathbb{R}.
\end{equation}
In local fibered coordinates,
\begin{equation}
 f^{T}(q^i,\dot{q}^i)=\dot{q}^i\frac{\partial f}{\partial q^i}.
\end{equation}

Given a tangent vector $v_q\in T_qQ$
with components $(q^i,v^i)$, we define its vertical lift $v_q^{V}$ for a point $w_q\in T_qQ$ by
\begin{equation}
 v_q^{V}=\frac{d}{dt}|_{t=0} \left(w_q+tv_q\right)
\end{equation}
If $X$ is vector field on $Q$, then its vertical lift to $TQ$ is the vector field
\begin{equation}
 X^{V}(w_q)=(X(q))_{w_q}^V\in T_{w_q}(TQ)
\end{equation}
for all $w_q$ in $TQ$.
Now, consider the flow $\phi_t:Q\rightarrow Q$ of a vector field $X$. We define the complete lift $X^{T}$ of $X$ to $TQ$ as the generator
of the tangent flow $T\phi_t:TQ\rightarrow TQ$. We are assuming, for simplicity, that the flow
generated by $X$ is complete, but the construction is still valid in general.
If $X=X^i\frac{\partial}{\partial q^i}$, a direct computation shows that
\begin{equation}
 X^{T}=X^i\frac{\partial}{\partial q^i}+\dot{q}^i\frac{\partial X^i}{\partial q^i}\frac{\partial}{\partial \dot{q}^i}.
\end{equation}
Consider now a $k$-form $\omega$ on $Q$. We define its complete lift $\omega^T$ to $TQ$ a $k$-form characterized by:
\begin{equation}
 \omega_Q(X_1^{T},\dots,X_n^{T})=\omega_Q(X_1,\dots,X_n)^{T}
\end{equation}
The following identity follows from a direct computation $d(\omega^{T})=(d\omega)^{T}$. We are particulary interested
in the case of lifts of symplectic forms on $Q$. Therefore, if $\omega_Q$ is symplectic on $Q$, then $\omega_Q^{T}$ is a symplectic form
on $TQ$. Indeed,
\begin{equation}
 \text{rank}(\omega_Q^{T})=2\ \text{rank}(\omega_Q)
\end{equation}
and $d(\omega^{T})=(d\omega)^{T}=0$.

\subsection{Submanifolds of symplectic manifolds}
Let $(M,\omega)$ be a symplectic manifold, and $N$ be a submanifold of $M$. We define the symplectic orthogonal complement
of $TN$ as the set of tangent vectors
\begin{equation}
 TN^{\bot}=\{u\in TM |\quad \omega(u,v)=0, \forall v\in TN\}.
\end{equation}
Note that, the dimension of the tangent bundle $TM$ is the sum of the tangent bundle $TN$ and its symplectic orthogonal complement $\ TN^{\bot}$. We list some of the important cases. 
 \begin{itemize}
  \item $N$ is called an isotropic submanifold of $M$ if $TN\subset TN^{\bot}$. In this case, the dimension of $N$ is less or equal to the half of the dimension of $M$.
  \item $N$ is called a coisotropic submanifold of $M$ if $TN^{\bot}\subset TN$. In this case, the dimension of $N$ is greater or equal to the half of the dimension of $M$.
  \item $N$ is called a Lagrangian submanifold of $M$ if $N$ is a maximal isotropic subspace of $(TM,\omega)$. That is, if $TN=TN^{\bot}$. In this case, the dimension of $N$ is equal to the half of the dimension of $M$. 
  \item $TN$ is symplectic if $TN\cap TN^{\bot}=0$. In this case, $(N,\omega_N)$ is a symplectic manifold.
 \end{itemize}

A diffeomorphism between two symplectic manifolds is called a symplectomorphism if it preserves the symplectic structures. Under a symplectomorphism, the image of a Lagrangian (isotropic, coisotropic, symplectic) submanifold is Lagrangian (resp. isotropic, coisotropic, symplectic) submanifold.

\subsection{The complete lift of $\omega_Q$} \label{TSS}

Consider the canonical symplectic manifold $(T^{\ast}Q,\omega_Q)$. The tangent bundle $TT^{\ast}Q$ of $T^{\ast}Q$ carries a symplectic two-form $\omega_{Q}^T$ that derives from two potential one-forms, denoted by $\theta_1=i_{T}\omega _{Q}$ and $\theta_2=d_{T}\theta _{Q}$. These one-forms are defined by the canonical forms $\omega_{Q}$ and $\theta_{Q}$, respectively. The definition of the derivation $i_{T}$ is the manifestation of the double vector bundle structure of $TTQ$ on $TQ$, and explicitly given by
\begin{equation*}
  i_{T}\omega _{Q}(X)=\omega _{Q}(\tau_{T^{\ast}Q}(X),T\tau_Q (X)).
\end{equation*}
The derivation $d_{T}$ is the commutator $[d,i_{T}]$.

Accordingly, in the local coordinate chart $(q^i,p_i;\dot{q}^i,\dot{p}_i)$, the potential one-forms are computed to be
\begin{equation}
\theta_{1}=i_{T}\omega _{Q}=\dot{p}_i dq^i-\dot{q}^i dp_i, \quad \theta_{2}=d_{T}\theta _{Q}=\dot{p}_i dq^i+p_i d\dot{q}^i  \label{thets}.
\end{equation}
The exterior derivatives of these one-forms are the same and define the symplectic two-form
\begin{equation} \label{Tstf}
\omega^T_Q=d\dot{p}_i \wedge dq^i+dp_i\wedge d\dot{q}^i
\end{equation}
is known as the complete lift to the tangent space of the symplectic two-form \cite{Tu77,Tu89}. Note that, the difference $\theta_{2}-\theta_{1}$ is an exact one-form. Actually, it is the exterior derivative of coupling function of the Legendre transformation between the Lagrangian and Hamiltonian formalisms. Further, existence of two potential one-forms for $\omega_Q^T$ leads to the existence of a Tulczyjew's triple, which is exhibited in the following subsection.

\subsection{Tulczyjew's triple}

We now consider a particular kind of symplectic manifolds introduced by Tulczyjew in \cite{Tu76,Tu77,Tu80,TuUr96,TuUr99,Tu89}. 

\medskip

A special symplectic manifold is a quintuple $(R,N,\tau,\theta,A)$ where $\tau:R\rightarrow N$ is a fiber bundle, $\theta$ is a one-form
on $R$ and $A:R\rightarrow T^{*}N$ is a diffemorphism such that $\pi=\pi_N\circ A$ and $\theta=A^{*}\theta_N$. Since $(T^{*}N,\omega_N=-d\theta_N)$ is a symplectic manifold, then $(R,\omega=-d\theta)$ is symplectic too and $A^{*}\omega_N=\omega$,
therefore, $(R,\omega)$ and $(T^{*}N, \omega_N)$ are symplectomorphic. Consider the following diagram.
\begin{equation} \label{sss}
\xymatrix{R \ar[dr]_{\tau} \ar[rr]^{A}&&T^{\ast }N \ar[dl]^{\pi_{N}}
\\&N}  
\end{equation}
Tulczyjew's symplectic space $(TT^*Q,\omega_Q^T)$ admits two special symplectic structures. Let us study them.

The non-degeneracy of the canonical symplectic structure $\omega_Q$ on $T^*Q$ leads to the existence of the following diffeomorphism
\begin{equation}\label{beta1}
 \beta_Q:TT^{*}Q \mapsto T^{*}T^{*}Q:X \mapsto \iota_X\omega_Q.
 \end{equation} 
The mapping $\beta_Q$ is actually a symplectomorphism if the iterated cotangent bundle $T^{*}T^{*}Q$ is equipped with the canonical symplectic two-from $\omega_{T^*Q}$.
In coordinates, we have that
 \begin{equation}\label{beta2}
  \beta_Q(q^i,p_i,\dot{q}^i,\dot{p}_i)=(q^i,p_i,-\dot{p}_i,\dot{q}^i).
 \end{equation}
It is a matter of a direct calculation to prove that, the quintuple $$(TT^{*}Q,T^{*}Q,\tau_{T^{*}Q},\theta_1,\beta_Q)$$ is a special symplectic manifold. Here, $\theta_1$ is the differential one-form defined in (\ref{thets}). 

We define a canonical diffeomorphism $\alpha_Q:TT^{*}Q\rightarrow T^{*}TQ$ as follows. First of all,
let us recall that there exists a canonical involution $S_Q:TTQ\rightarrow TTQ$ locally given by
\begin{equation}
 S_Q(q^i,\dot{q}^i,\delta q^i, \delta \dot{q}^i)=(q^i,\delta q^i, \dot{q}^i, \delta \dot{q}^i)
\end{equation}
(see \cite{godbillion}). Now, given $v\in TT^{*}Q$ we define $\alpha_Q(v)\in T^{*}TQ$ as
\begin{equation}
 \langle w,\alpha_Q(v)\rangle=\frac{d}{dt}\langle \gamma, \xi \rangle|_{t=0}
\end{equation}
where $\gamma:\mathbb{R}\rightarrow TQ$ and $\xi:\mathbb{R}\rightarrow T^{*}Q$ are curves such that
$j^{1}\circ \gamma=S_Q(w)$, $j^{1}\circ \xi=v$ and $\tau_Q\circ \gamma=\pi_Q\circ \xi$. In local coordinates, we obtain
\begin{equation}\label{defalpha}
 \alpha_Q(q^i,p_i,\dot{q}^i,\dot{p}_i)=(q^i,\dot{q}^i, \dot{p}^i, p_i).
\end{equation}
The mapping $\alpha_Q$ is a symplectomorphism if the iterated cotangent bundle $T^{*}TQ$ is equipped with the canonical symplectic two-from $\omega_{TQ}$.
Then it becomes immediate to prove that  $$(TT^{*}Q,TQ,T\pi_Q,\theta_2,\alpha_Q)$$ is a special symplectic manifold. Here, $\theta_2$ is the differential one-form defined in (\ref{thets}).

As a result, we have derived two special symplectic structures for the symplectic manifold $(TT^*Q,\omega_Q^T)$. Tulczyjew's triple is the combination of these two special symplectic structures in one commutative diagram as given below.
\begin{equation}
\xymatrix{T^{\ast }TQ \ar[dr]^{\pi_{TQ}}&&TT^{\ast }Q\ar[dl]^{T\pi_{Q}}
\ar[rr]^{\beta_{Q}} \ar@<1ex>[dr]^{\tau_{T^{\ast }Q}}
\ar[ll]_{\alpha_{Q}}&&T^{\ast }T^{\ast }Q\ar[dl]^{\pi _{T^{\ast}Q}}
\\&TQ&&T^{\ast}Q } \label{T}
\end{equation}

\section{Lagrangian Submanifolds}
In this and following subsections, we will summarize the theory of Morse families and Lagrangian submanifolds generated by them. We refer an incomplete list \cite{Be11, LiMa, Tu80,
TuUr96, TuUr99, We77} for more detailed discussions.

 Let $(M,\omega)$ a symplectic manifold. A sufficient condition for a submanifold $N\subset M$ to be Lagrangian is $TN=TN^{\bot}$.
 If $N$ is an isotropic subspace of a symplectic manifold $(M,\omega)$, then $N$ is Lagrangian if an only if $\text{dim}N=\text{dim}M/2$. For different types of manifolds (Poisson, Nambu--Poisson, etc), the definition of a Lagrangian submanifold has been accommodated to its background. See for example \cite{leonsardon1,leonsardon2,LiMa}.

Two principal examples of Lagrangian submanifolds of the symplectic phase space $T^{*}Q$ are the fibers
of the canonical cotangent bundle projection and the image space of closed one-forms $\gamma:Q\rightarrow T^{*}Q$.
The latter case includes the zero section of the projection $T^*Q\mapsto Q$ as well. This statement is the well-known
Weinstein tubular neighborhood theorem \cite{We77}.

\subsection{Morse families}

Consider a differentiable fibration $(P,\pi,N)$. A real valued function $F$ on the total space $P$
can intuitively be understood as a family of functions parameterized by the coordinates of the fibers of $\pi$. The critical set of $F$ is defined by
\begin{equation}
Cr\left( F,\pi \right) = \{ z \in P:\left\langle
dF(z) ,V\right\rangle =0,\forall V \in V_zP \}
\end{equation}%
and is a submanifold of $P$. The dimension of $Cr\left(F,\pi \right)$ is equal to the dimension of $N$.
Here, $VP$ is the vertical bundle on $P$ consisting of vertical vectors projecting to the zero section of $TN$
under the mapping $T\pi$. We define a bilinear mapping
\begin{eqnarray} \label{W}
W\left( F,z\right) &:&V_z P\times
T_z P\rightarrow
\mathbb{R}
\notag \\
&:&\left( v,w \right) \rightarrow
D^{\left( 1,1\right) }\left( F\circ \chi \right) \left( 0,0\right) ,
\end{eqnarray}%
where $\chi :%
\mathbb{R}
^{2}\rightarrow P$ is the mapping such that the vector $v$ is obtained by taking the derivative of $\chi$ with respect to its first entry at $(0,0)$ and the vector $w$ is obtained by taking the derivative of $\chi$ with respect to its second entry at $(0,0)$.

A family of functions $F$ defined on the total space of the fibration $(P,\pi,N)$ is said to be regular if the rank of the matrix $W\left( F,z
\right) $ defined in (\ref{W}) is the same at each $z\in Cr\left( F,\pi \right)
.$ A family of functions $F$ is said to be a Morse family (or an energy function) if the rank of $W\left( F,z
\right) $ is maximal at each $z\in Cr\left( F,\pi \right)$.

Let us write the requirement of being a Morse family in terms of local coordinates. Assume that the dimension of the manifold $N$ is $n$ with coordinates $(q^i)$, the dimension of a fiber is $k$ with coordinates $(\lambda^a)$. The function $F$ is called a Morse family if the rank of the $n\times (n+k)$-matrix%
\begin{equation}
\left( \frac{\partial ^{2}F}{\partial q^{i}\partial q^{j}}\text{ \ \ }\frac{%
\partial ^{2}F}{\partial q^{i}\partial \lambda^{a}}\right)  \label{MorseCon}
\end{equation}%
is maximal.

\subsection{Lagrangian submanifolds generated by Morse families} \label{LS-MF}

A Morse family $F$ on the smooth bundle $\left( P,\pi ,N\right) $
generates a Lagrangian submanifold
\begin{equation}
S_{T^{\ast }N}=\left\{ w \in T^{\ast }N:T^{\ast }\pi (w
)=dF\left( z\right) \right\}  \label{LagSub}
\end{equation}%
of $\left(
T^{\ast }N,\omega _{N}\right)$. In this case, we say that $S_{T^{\ast }N}$ is generated by the Morse family $F$. Note that, in the definition of $%
S_{T^{\ast }N}$, there is an intrinsic requirement that $\pi \left( z\right)
=\pi _{T^{\ast }N}\left( w \right) $. Here, we are presenting the following diagram in order to summarize the discussion. 
 \begin{equation} \label{Morse-pre}
  \xymatrix{
\mathbb{R}& P \ar[d]^{\pi}\ar[l]^{F}& T^*N \ar[d]^{\pi_N}\\ &
N  \ar@{=}[r]& N
}
\end{equation}
In order to exhibit the structure of the submanifold $S_{T^{\ast }N}$, define a fiber preserving mapping $\kappa $ from the critical set $Cr\left( F,\pi \right)$ of the Morse family $F$ to the cotangent bundle $T^{\ast
}N$ according to the requirement
\begin{equation}
\left\langle \kappa (z) , Z_{N}\right\rangle =\left\langle
dF,Z_{P}\right\rangle ,
\end{equation}%
which is valid for all $\pi$-related vector fields $Z_{N}\in \mathfrak{X}(N)$ and $Z_{P}\in \mathfrak{X}(P)$. Note that, this mapping is an immersion and that $\dim \left( S_{T^{*}N}\right)$ equals to $\dim \left( Cr\left( F,\pi \right) \right)=n$. A direct calculation shows that the image space of $\kappa$ is the Lagrangian submanifold $S_{T^{\ast }N}$ defined in (\ref{LagSub}). 

Let $N$ be an immersed submanifold of $Q$, and $T_{N}^{\ast }Q$ denote the inverse image $\pi
_{Q}^{-1}\left( N\right)$ in the cotangent bundle $T^{\ast }Q$. We define the
mapping
\begin{equation}
\xi :T_{N}^{\ast }Q\rightarrow T^{\ast }N
\end{equation}%
by requiring that the equality $$\left\langle \xi \left( p\right) ,Z_N(n)\right\rangle =\left\langle
p,Z_N(n)\right\rangle $$ holds for each $Z_N\in \mathfrak{X}(N)$, where $n=\pi_Q(p)$ and $p\in T_{N}^{\ast }Q$. So, $\xi$ is the identity
if $p\in N,$ that is, if $p\in T^{\ast
}N.$ Consider the canonical injection $i:T_{N}^{\ast
}Q\rightarrow T^{\ast }Q$. If $S_{T^\ast N}$ is a Lagrangian submanifold of $T^{\ast }N$ then the preimage $i\circ \xi ^{-1} (S_{T^*N})$ is a Lagrangian submanifold $S_{T^*Q}$ of $T^{\ast }Q$. If, particularly, the Lagrangian submanifold $S_{T^\ast N}$ is a generated by a Morse family $F$ on the fiber bundle $(P,\pi,N)$ then, using the same terminology, we say that $S_{T^*Q}$ is generated by the Morse family $F$. We are presenting the following diagram.
 \begin{equation} \label{Morse}
  \xymatrix{
\mathbb{R}& P \ar[d]^{\pi}\ar[l]^{F}& T^*Q \ar[d]^{\pi_Q}\\ &
N  \ar@{^{(}->}[r]& Q
}
\end{equation}
Let us try to see this more explicitly in the following way. For every point $p\in T^{\ast }Q$ such that $\pi _{N}\left( p\right)\in N$, we can find a point $z\in P$ satisfying $\pi \left( z\right) =\pi _{N}\left( p\right) $.
Then, for every vector $Z_{P}(z)$, we have that $T\pi \circ
Z_{P}$ is a vector field on $N$, and hence a vector field on $Q$. Then, the elements of the Lagrangian submanifold $S_{T^*Q}$ are defined by the requirement $$\left\langle
p ,T\pi \circ
Z_{P}(z)\right\rangle =\left\langle dF(z),Z_{P}(z)\right\rangle,$$
where $F$ is a Morse family on $P$.

If the local coordinates $(q^i,\lambda^a)$ are considered, then the Lagrangian submanifold generated by the Morse family $F$ is defined by 
\begin{equation}
S_{T^*Q}=\left \{ \left(q^i,\frac{\partial F}{ \partial q^i} (q,\lambda)\right)\in T^*Q: \frac{\partial F}{ \partial \lambda^a} (q,\lambda)=0  \right \}.
\end{equation}
Note that, we are not distinguishing here the base manifold $N$ from $Q$.

\subsection{Lagrangian submanifolds of special symplectic structures}

Let $\left(R,N=Q,\tau ,\omega=-d\theta ,A \right) $ be a special symplectic
structure for a symplectic manifold $\left(R,\omega \right) $ with symplectomorphism $A$. If $S_{T^*Q}$ is a Lagrangian submanifold of $T^*Q$, then its pre-image $S=A^{-1}(S_{T^*Q})$ under the symplectomorphism $A$ is a Lagrangian submanifold of the symplectic manifold $(R,\omega)$. If the Lagrangian submanifold $S_{T^*Q}$ is generated by the Morse family $F$ as presented in the diagram (\ref{Morse}), then we say that $S$ is generated by the Morse family $F$ as well. To illustrate this, we draw the following diagram.
 \begin{equation} \label{Morse-Gen}
 \xymatrix{
\mathbb{R}& P \ar[d]^{\pi}\ar[l]^{F}& T^*Q \ar[d]_{\pi_Q}& R \ar[l]_{A} \ar[dl]^{\tau}\\ &
N \ar@{^{(}->}[r]& Q
}
\end{equation}

Let us record here some special cases of the diagram for future reference. 

 \begin{itemize}
\item The simplest case occurs if $A$ is the identity mapping (no special symplectic manifold) on $T^{\ast }Q$, $P=N$ (no Morse family), and the submanifold $N=Q$ (no constraints on $Q$). Then, we have a function $F$ (not a family) on $Q$ and its exterior derivative is a Lagrangian submanifold of $T^*Q$. Then, $S$ turns out to be the image space of a closed one-form $dF:Q\rightarrow T^{*}Q$.
This case includes the zero section of the projection $T^*Q\mapsto Q$ as well \cite{We77}.

\item If $A$ is the identity mapping on $T^{\ast }Q$ (no special symplectic manifold), $P=N$ (no Morse family). Then we have a function $F$ (not a family) on a submanifold of $N$ of $Q$, and the Lagrangian submanifold
\begin{equation}
 S_{T^*Q}=\{p\in T^*Q : \pi_Q(p)\in N, \langle Z,\theta_Q (p)\rangle  =  \langle T\pi(Z), dF\rangle\}
\end{equation}
for any $Z\in T_pT^*N$ such that $T\pi(Z)\in TN\subset TQ$, \cite{leonrodrigues}.

\item Let $P=N$ (no Morse family) and the submanifold $N=Q$ (no constraints on $Q$). Instead, consider the existence of a non-trivial special symplectic structure $\left( P,N,\pi ,\omega ,A \right) $. Then, a Lagrangian submanifold $S$ of $(P,\omega)$ is defined by the pre-image of $dF$ under the isomorphism $A$, that is
\begin{equation}
S=A^{-1}\left( dF\right)
=\left\{y\in Y:\left\langle y,u\right\rangle =\left\langle dE,T\tau(u)
\right\rangle ,\forall u\in T_{y}Y\right\}
\end{equation}%
where $E$ is defined on $Q$.

\item Let $P=N$ (no Morse family), $N$ be a proper submanifold of $Q$, $(R,N,\tau,\theta,A)$ be a special symplectic manifold. Then the
set
\begin{equation}
S=\left\{ y\in R: \tau \left( y\right) \in N,\left\langle \theta
,u\right\rangle =\left\langle dF,T\tau (u) \right\rangle ,\forall u\in
T_{y}Y\right\}
\end{equation}%
is a Lagrangian submanifold of $\left( R,-d\theta \right) $, and said to be
generated by the function $F:N\rightarrow
\mathbb{R}
.$ We cite chapter 7 of \cite{leonrodrigues} for a proof of this statement in a more general framework. 

\end{itemize}

\subsection{Lagrangian submanifolds of Tulczyjew's symplectic space} \label{LSTT*Q}

Assume that a submanifold $E$ of $TT^{*}Q$ is defined in terms of the constraint functions $\Phi^{A}:TT^{*}Q\rightarrow \mathbb{R}$,
  $$\Phi^A(q^i,p_i,\dot{q}^i,\dot{p}_i)=0. $$ 
If $E$ is an IHS, then it must be a Lagrangian submanifold of
the symplectic space $(TT^{*}Q,\omega_Q^{T})$ and the number of constraints must be $2n$ assuming that the dimension of $Q$ is $n$.
Note that, in this case, the Poisson brackets of the constraint functions must vanish \cite{LiMa}, 
  \begin{equation}\label{LSC}
    \{\Phi^A,\Phi^B\}=0.
  \end{equation}
Here, the Poisson bracket in (\ref{LSC}) is the one induced by the
  Tulczyjew's symplectic two-form  $\omega_Q^T$ in the
  subsection (\ref{TSS}). This Poisson bracket can be computed as 
    \begin{equation} \label{PB}
    \{f,g\}=\omega_Q^{T}(X_f,X_g)=\frac{\partial f}{\partial \dot{p}_i}\frac{\partial g}{\partial q^i}-\frac{\partial g}{\partial \dot{p}_i}\frac{\partial f}{\partial q^i}+\frac{\partial f}{\partial p_i}\frac{\partial g}{\partial \dot{q}^i}-\frac{\partial g}{\partial p_i}\frac{\partial f}{\partial \dot{q}^i}.
   \end{equation}

The image of a Hamiltonian vector field is a Lagrangian submanifold of $TT^*Q$. Conversely, if a Lagrangian submanifold of $TT^*Q$ is the image of a vector field on $T^*Q$, then this vector field is a (at least locally) Hamiltonian vector field. To see this, we present the following calculations. Assume that $E$ is a Lagrangian submanifold of $TT^*Q$ and there exists a vector field $X$ satisfying $E=\text{Im}(X)$. In Darboux's coordinates, we write $X$ as
\begin{equation}
 X=\phi^i(q,p)\frac{\partial}{\partial q^i}+\phi_i(q,p)\frac{\partial}{\partial p_i},
\end{equation}
where $\phi^i$ and $\phi_i$ are arbitrary functions on $T^*Q$.  This local picture enables us to write the fiber variables $(\dot{q},\dot{p})$ in terms of the functions of the base variables $(q,p)$, and the following definition of the Lagrangian submanifold of $E$ given by
    \begin{equation}\label{Eedef}
E=\left \{(q^i,p_j;\dot{q}^i, \dot{p}_i)\in TT^*Q: \dot{q}^i-\phi^i(q,p)=0, \dot{p}_i+\phi_i(q,p)=0\right \}. 
 \end{equation}
Since $E$ is a Lagrangian submanifold, the Poisson brackets of the defining equations in (\ref{Eedef}) must be identically zero. Here, the Poisson bracket is the one in \eqref{PB}. This requirement dictates $n^2$ number of conditions
\begin{equation*}
  \frac{\partial \phi_j}{\partial p_i}-\frac{\partial \phi^i}{\partial q^j}=0.
\end{equation*}
These conditions are the same with the conditions for closure of the one-form $\phi=\phi_j dq^j+\phi^i dp_i$.
Locally, every closed one-form is exact. This implies that there exists a Hamiltonian function $H$ depending on $(q,p)$ satisfying $dH=\phi$. As a result, the system (\ref{Eedef}) can be written as in form of the Hamilton's equations (\ref{HamEq}). 

A stronger result follows from the generalized Poincar\'{e} lemma \cite{BeTu80,Ja00,LiMa,TuUr99,We77}. The generelized Poincar\'{e} guarantees that for a Lagrangian submanifold $E$ of $TT^*Q$, whether it is explicit or implicit, there exists a Morse family of functions generating $E$. 
This theorem is also known as Maslov-H\"{o}rmander theorem \cite{BeCaSi09,Be11,Ca91,Ca15}. We record here a Morse family generating the Lagrangian submanifold $E$ as follows 
 \begin{equation} \label{Morse-Gen}
 \xymatrix{
\mathbb{R}& P \ar[d]^{\pi}\ar[l]^{F}& T^*T^*Q \ar[d]_{\pi_{T^*Q}}&TT^*Q \ar[l]_{\beta_Q} \ar[dl]^{\tau_{T^*Q}}\\ &
N \ar@{^{(}->}[r]& T^*Q
}
\end{equation}
where $N$ is a submanifold of the cotangent bundle $T^*Q$, and $F$ is a Morse family defined on the total space $P$ of the smooth bundle $(P,\pi,N)$. On the local chart, the Lagrangian submanifold $E$ generated by $F$ is computed to be 
\begin{equation} \label{MFGen}
E=\left \{ \left(q^i,p_i;\frac{\partial F}{\partial p_i},-\frac{\partial F}{\partial q^i}\right)\in TT^*Q: \frac{\partial F}{\partial \lambda^a}(q,p,\lambda)=0\right \}.
\end{equation} 
Note that, if the Morse family $F$ does not depend on the fiber variables $(\lambda)$, then $E$ becomes explicit. 

Let us comment on a particular case. Consider the following constrained Hamiltonian (Dirac) system
\begin{equation}\label{const}
 \dot{q}^i=\frac{\partial H}{\partial p_i}(q,p),\quad \dot{p}_i=-\frac{\partial H}{\partial q^i}(q,p),\quad \Phi^{\alpha}(q,p)=0,
\end{equation}
\noindent
where $\Phi^{\alpha}$, for $\alpha=1,...,k$, are real valued functions defining a constraint submanifold $M\subset T^*Q$ \cite{MaMeTu97}. Note that, we prefer to denote the generator by $H$ to highlight that it is a Hamiltonian function in the classical sense. Diagrammatically, we have the following picture generating the dynamics
\begin{equation} \label{cHs}
 \xymatrix{
\mathbb{R}& \ar[l]_{F}P \ar[d]& T^*T^*Q \ar[d]_{\pi_{T^*Q}}&TT^*Q \ar[l]_{\beta_Q} \ar[dl]^{\tau_{T^*Q}}\\ & 
T^*Q \ar@{=}[r]& T^*Q
}
\end{equation}
where $P$ is the product manifold $T^*Q\times \mathbb{R}^k$ and the Morse family is given by $F(q,p,\lambda)=H(q,p)+\lambda_\alpha\Phi^{\alpha}(q,p)$.

\section{Hamilton--Jacobi theory for implicit systems}

\subsection{Geometry of the Hamilton--Jacobi equation}

A HJE is a partial differential equation for a generating function $S(q^i,t)$ on $Q$ and the time $t$ given by
\begin{equation}\label{tdepHJ}
 \frac{\partial S}{\partial t}+H\left(q^i,\frac{\partial S}{\partial q^i}\right)=0.
\end{equation}
Note that, the generalized momenta do not appear in (\ref{tdepHJ}), except as derivatives of $S$.
This equation is a necessary condition describing the external geometry in problems of calculus of variations. Hamilton's principal function
$S=S(q^i,t)$, which is the solution of the HJE, and the classical function $H$ are both closely related to the classical action
\begin{equation*}
 S=\int{L dt}.
\end{equation*}

The function $S$  is a generating function
for a family of symplectic flows that describes the dynamics of the Hamilton equations.
 If the generating function is separable in time, then we can make an ansatz $S(q^i,t)=W(q^i)-Et,$
where $E$ is the total energy of the system. Then, HJE in \eqref{tdepHJ} reduces to
 \begin{equation}\label{hje}
H\left(q^i,\frac{\partial W}{\partial q^i}\right)=E.
 \end{equation}
Physically, this constant is identified with the energy of the mechanical system.

Let us summarize the geometric Hamilton-Jacobi theory. For this, first consider a Hamiltonian vector field $X_H$ on $T^{*}Q$, and a one-form section $\gamma$ on $Q$.
We define a vector field $X_H^{\gamma}$ on $Q$ by
\begin{equation}\label{gammarelated}
 X_H^{\gamma}=T\pi\circ X_H\circ \gamma.
\end{equation}
This definition implies the commutativity of the following diagram.
\begin{equation}\label{Xg}
  \xymatrix{ T^{*}Q
\ar[dd]^{\pi_{Q}} \ar[rrr]^{X_H}&   & &TT^{*}Q\ar[dd]^{T\pi_{Q}}\\
  &  & &\\
 Q\ar@/^2pc/[uu]^{\gamma}\ar[rrr]^{X_H^{\gamma}}&  & & TQ }
\end{equation}
We enunciate the following theorem.

\begin{theorem} \label{HJT}
The closed one-form $\gamma=dW$ on $Q$ is a solution of the Hamilton--Jacobi equation \eqref{hje} if the following conditions are satisfied:

\begin{enumerate}
\item The vector fields $X_{H}$ and $X_{H}^{\gamma }$ are $\gamma$-related, that is

\begin{equation}
T\gamma(X^{\gamma})=X\circ\gamma.
\end{equation}
\item Or equivalently, if the following equation is fulfilled
$$d\left( H\circ \gamma \right)=0.$$
\end{enumerate}
\end{theorem}
\begin{proof}
We refer \cite{CaGrMaMaMuRo06,deLeIgdeDi08,deLeMadeDi10}
for proof of this theorem.
\end{proof}
\noindent
The first item in the theorem says that if $\left( q^i\left( t\right) \right) $
is an integral curve of $X_{H}^{\gamma }$, then $\left(
q^i\left( t\right) ,\gamma_j \left( q\left( t\right) \right) \right) $
is an integral curve of the Hamiltonian vector field $X_{H}$, hence a
solution of the Hamilton's equations \eqref{HamEq}. Such a solution of the Hamiltonian equations is
called horizontal since it is on the image of a one-form on $Q$. In the local picture, the second
condition implies that exterior derivative of the Hamiltonian function on the image of $\gamma $ is closed, that is, $%
H\circ \gamma $ is constant 
\begin{equation}
H\left( q^{i},\gamma _{j}\left( q\right) \right) = \text{cst}.  \label{HJ-0}
\end{equation}%
Under the assumption that $\gamma $ is closed, we can find (at least locally) a
function $W$ on $Q$ satisfying $dW=\gamma $. After the substitution of this,
equation (\ref{HJ-0}) retrieves the HJE (\ref{hje}).

%

\subsection{Implicit differential equations}
A first-order  differential equation on a manifold $M$ is a submanifold of its tangent bundle. The submanifold is said to be
an explicit differential equation (EDE) if it is image of a vector field defined on $M$, otherwise, it is called
implicit (IDE) \cite{Ja00,MaMeTu92,MaMeTu97,Ta76,MeMaTu95}.
Consider a local coordinate system $(q^i)$ on $M$, and the induced coordinates $(q^i,\dot{q}^i)$ on
its tangent bundle $TM$. An IDE can be written in form
\begin{equation} \label{ide}
  \dot{q}^i=g^i(q,\lambda), \qquad f^a(q,\lambda)=0
\end{equation}
for $i$ with certain values running from $1,...,n$ and $a=1,...,k$. Here, $g^i$ and $f^a$ are real valued differentiable functions on $(q,\lambda)\in M\times \mathbb{R}^k$.
Note that $f^{a}$'s form a matrix of maximal rank
\begin{equation}
 \text{rank}\left(\frac{\partial f^a}{\partial q^i},\frac{\partial f^a}{\partial \lambda^b}\right)=k.
\end{equation}
This implies that only certain $\dot{q}^i$'s are expressible as in \eqref{ide} depending on certain $(q^j,\lambda^a)$ for some $a$'s
equal to $k\leq n$.


A solution  of an IDE is a curve $\phi $ on $M$ satisfying that the
tangent vectors $\dot{\phi}\left( t\right)$ belong to $E$
for all $t$. The submanifold $E$ is called {\it integrable} if for all vectors $v\in E$, there
exists a solution $\phi$ satisfying $v=\dot{\phi}\left( t\right) $ for some
$t$. Equivalently, we may define ``integrability'' of an IDE without refering to the solution curves as follows. We say that the IDE
is integrable if the restriction of the tangent bundle projection $\tau_{M}$ to $E$
is surjective summersion and if $E\subset T\left( \tau _{M}\left( E\right)
\right) $.

In principle, an IDE is not necessarily integrable. For example, consider the following case \cite{MeMaTu95}.
\begin{example}
 Let $E\subset TT^{*}Q$ be a system of IDE, where $Q$ is coordinated by $q$ and $T^{*}Q$ by $(q,p)$, defined by
 \begin{equation}
  E=\{(q,p,\dot{q},\dot{p})\in TT^{*}Q, q^2+p^2+(\dot{q}+1)^2+\dot{p}=k\}
 \end{equation}
is not integrable in points $q^2+p^2=k$, $\dot{q}=-1$ and $\dot{p}=0$ with $q\neq 0$.
\end{example}

Nonetheless, we
can develop an algorithm to extract its integrable part \cite{MeMaTu95}. This algorithm
works as follows. We denote the projection of the submanifold $E$ onto $M$ by $C$%
. Each step of the algorithm, there is a fiber bundle $\left( E^{k},\tau
^{k},C^{k}\right) $ consisting of two submanifolds $E^{k}\subset E$ and $%
C^{k}\subset C$, and a (surjective submersion) projection $\tau
^{k}:E^{k}\rightarrow C^{k}$. The first step is initiated by choosing $%
\left( E^{0}=E,C^{0}=C\right) $ and, iteratively, the further steps are
defined by%
\begin{equation*}
\left(
\begin{array}{c}
E^{1}:=E^{0}\cap TC^{0} \\
C^{1}:=\tau _{Q}\left( E^{1}\right) \\
\tau ^{1}:E^{1}\rightarrow C^{1}%
\end{array}%
\right) \rightarrow ...\rightarrow \left(
\begin{array}{c}
E^{k}:=E^{k-1}\cap TC^{k-1} \\
C^{k}:=\tau _{Q}\left( E^{k}\right) \\
\tau ^{k}:E^{k}\rightarrow C^{k}%
\end{array}%
\right) \rightarrow ...\text{ }.
\end{equation*}%
In finite dimensions, there is an end for the algorithm, that is, there
exists a three tuple $\left( E^{f},\tau ^{f},C^{f}\right) $ satisfying that $%
E^{f+1}=E^{f}$ and $C^{f+1}=C^{f}$. Note that, the final manifold $E^{f}$ is
integrable. We call $E^{f}$ the integrable part of $E$. Constrained 


\begin{example}\label{examprs}
 Consider the following set $(x,p,r,s,\dot{x},\dot{p},\dot{r},\dot{s})$ on $T\mathbb{R}^4$, with the following equations \cite{MeMaTu95}
 \begin{equation}\label{nonintpart}
  E=\left\{r=p, s=0, \dot{r}=-\frac{\partial H}{\partial x}(x,p), \dot{s}=\dot{x}-\frac{\partial H}{\partial p}(x,p)\right \}
 \end{equation}
 Here, $E=E^{0}$ and $C^{0}=\{(x,p,r=p,s=0)\}$, such that $TC^{0}=\{\dot{x},\dot{p},\dot{r}=\dot{p},\dot{s}=0\}$.
 So,
 \begin{equation}\label{intpart}
  E^{0}\cap TC^{0}=\left\{ \dot{x}=\frac{\partial H}{\partial p}(x,p), \dot{r}=\dot{p}=-\frac{\partial H}{\partial x}(x,p), r=p, s=0 \right \}
 \end{equation}
From here, any other iteration $E^{k}=E^{1}$ and $TC^{k}=TC^{1}$. Hence, \eqref{intpart} is the integrable part of \eqref{nonintpart}.

\end{example}

\subsection{The Hamilton--Jacobi theory for implicit differential equations}

In this section we develop a geometric Hamilton-Jacobi theory for IDE.
Our problem is that given a set of IDE, we are not necessarily provided with a Hamiltonian vector field as explained in former sections.
Here, we propose two methods to construct our theory.
The first method consists of a theory which does refer to vector fields. The second is based on the construction of a local vector field defined on the image of
a section, but not globally on the phase space.

Let us start with the first method. Consider a submanifold $E$ of $TT^* Q$. By projecting $E$ by the tangent mapping $T\pi _Q $ onto the tangent bundle $TQ$, we arrive a submanifold $T\pi _Q (E)$ of $TQ$. Note that, $E$ refers to an IDE on $T^*Q$, whereas $T\pi _Q (E)$ refers to an IDE on $Q$. If $E$ is integrable, then $T\pi _Q (E)$ is integrable too. We see this by considering the projection $T{\pi_Q}(V)=v\in T\pi _Q (E)$ of an element $V\in E$. Note that, if $\varphi$ is a curve on $T^*Q$ and it is tangent to $V\in E$, then $ \pi_Q\circ \varphi$ is curve on $Q$ which is tangent to $v\in T\pi _Q (E)$. This shows that the projections of the solutions of $E$ are solutions of $T\pi _Q (E)$. 
Our aim is to discuss the inverse question, starting from the solutions of $T\pi _Q (E)$
construct solutions of $E$, that is to lift the solutions on $Q$ to the cotangent bundle $T^*Q$. This is the philosophy of a geometric Hamilton--Jacobi theory.
(Recall the geometric Hamilton--Jacobi theory exposed in subsection 4.1.)
Furthermore, if $E$ were not integrable, we would have to perform the integrability algorithm explained in subsection 4.2.

To answer the question we have to introduce a section $\gamma:Q\rightarrow T^*Q$.
If any solution $\psi:\mathbb{R}\rightarrow Q$ of $T\pi _Q (E)$ is a solution $\gamma\circ\psi:\mathbb{R}\rightarrow T^*Q$ of $E$,
then we denote the submanifold $T\pi _Q (E)$ by $E^\gamma$, and say that $E$ and $E^\gamma$ are $\gamma-$related.
We illustrate this in a diagram.

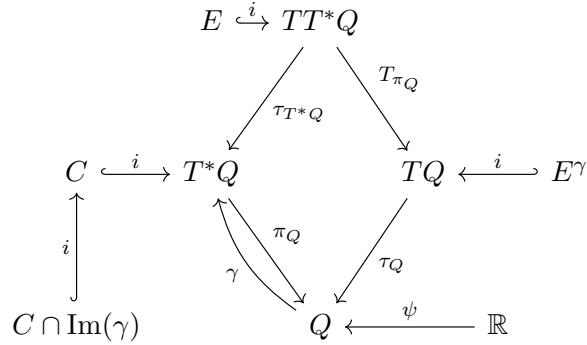
\begin{center}
 \begin{tikzcd}[column sep=tiny,row sep=huge]
  & &  &E\arrow[r, hook, "i"] &TT^*Q \arrow[ld, "\tau_{T^*Q}"] \arrow[rd, "T_{\pi_Q}"] & & \\
  & &C\arrow[r, hook, "i"] &T^*Q\arrow[dr,"\pi_{Q}"] & &TQ\arrow[dl,"\tau_{Q}"]  &&E^{\gamma}\arrow[ll, hook', "i"']\\
& &C\cap \text{Im}(\gamma)\arrow[u, hook, "i"] & &Q\arrow[lu,"\gamma",bend left=20] & &\mathbb{R}\arrow[ll,"\psi"']
 \end{tikzcd}
\end{center}

\medskip

In coordinates, a submanifold $E$ of $TT^{*}Q$ can be given by the set of functions $\Phi^{A}:TT^{*}Q\rightarrow \mathbb{R}$,
  $$\Phi^A(q^i,p_i,\dot{q}^i,\dot{p}_i)=0 $$  The projection of $E$ onto $T^{*}Q$ by means of $\tau_{T^* Q}$ results with a submanifold $C$ of $T^{*}Q$ given by the set of functions
  $$\Psi^{\alpha}(q^i,p_i)=0.$$
  Consider the intersection of the projected submanifold $C$ and the image space of the one-form $\gamma$. We denote this globally by $C\cap \text{Im}(\gamma)$ and locally by $\Psi^{\alpha}(q^i,\gamma_j(q))=0$. If a solution curve is represented by $\psi^i(t)\subset Q$, the composition $\gamma \circ \psi (t)=(\psi^i(t),\gamma_i(\psi(t)))$ is a curve on $T^{*}Q$ and the time derivative of the curve is
\begin{eqnarray*}
\frac{d}{dt}(\gamma \circ \psi)(0)&=&T\gamma (\psi(0)) \cdot \dot{\psi}(0) \\&=&\left(\psi^i(0),\gamma_i(\psi(0)),\dot{\psi}^i(0),\frac{\partial \gamma_j}{\partial q^i}\dot{\psi}^i(0)\right).
\end{eqnarray*}
 Then, the equations of the submanifold $E$ along $\gamma$ take the form
 \begin{equation}
  \Phi^{A}\left(q^i,\gamma_i(q),\dot{q}^i,\frac{\partial \gamma_j}{\partial q^i}\dot{q}^i\right)=0,
 \end{equation}
provided that $\Psi^{\alpha}(q^i,\gamma_i(q))=0$ along $\gamma(Q)\subset C$. Here, we assumed that $\psi^i(0)=q^i$.

 In the second case, we consider an additional section $\sigma:T^{*}Q\rightarrow TT^{*}Q$
 such that $\sigma(C)\subset E$.

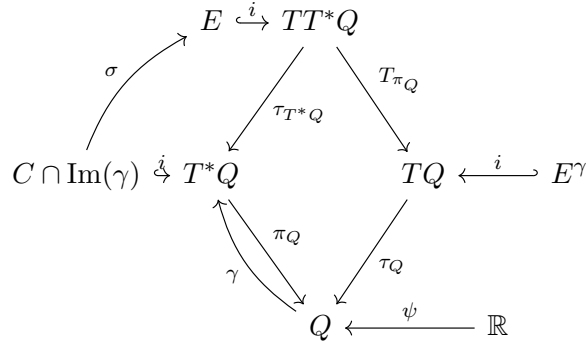
 \begin{center}
 \begin{tikzcd}[column sep=tiny,row sep=huge]
  & &  &E\arrow[r, hook, "i"] &TT^*Q \arrow[ld, "\tau_{T^*Q}"] \arrow[rd, "T_{\pi_Q}"] & & \\
  & &C\cap \text{Im}(\gamma)\arrow[ur,"\sigma", bend left=
20]\arrow[r, hook, "i"] &T^*Q\arrow[dr,"\pi_{Q}"] & &TQ\arrow[dl,"\tau_{Q}"]  &&E^{\gamma}\arrow[ll, hook', "i"']\\
& & & &Q\arrow[lu,"\gamma",bend left=20] & &\mathbb{R}\arrow[ll,"\psi"']
 \end{tikzcd}
\end{center}
Since $E$ is implicit, there may exist several vectors in $E$ projecting to the same point, say $c$, in $C$. The role of the section $\sigma$ is to reduce this unknown number to one. We are additionally require that the domain of the section $\sigma$ be the intersection of $\Ima(\gamma)$ and $C$ since, for implicit systems, $C$ may not be the whole of $T^*Q$. As a result, we arrive at a vector field $X_\sigma$. Note that $X_\sigma$ satisfies
\begin{equation}
 \iota_{X_{\sigma}}\omega_Q=\Theta(\gamma(q)).
\end{equation}
for an arbitrary one-form $\Theta$ defined on $\gamma(q)$.

We define a vector field $X_{\sigma}^{\gamma}$ on the tangent bundle $TQ$ by the commutation of the following diagram.
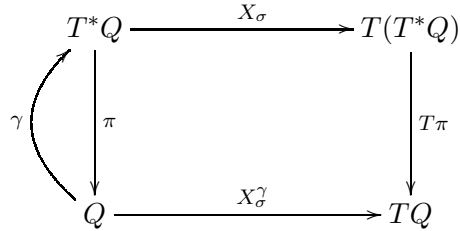
\[
\xymatrix{ T^{*}Q
\ar[dd]^{\pi} \ar[rrr]^{X_{\sigma}}&   & &T(T^{*}Q)\ar[dd]^{T\pi}\\
  &  & &\\
 Q\ar@/^2pc/[uu]^{\gamma}\ar[rrr]^{X_{\sigma}^{\gamma}}&  & & TQ }
\]
Explicitly,
\begin{equation}\label{hjsect}
 X_{\sigma}^{\gamma}=T_{\pi}\circ X_{\sigma}\circ \gamma.
\end{equation}

In local coordinates, the vector field $X_{\sigma}$ and its projection $X_{\sigma}^{\gamma}$ can be written as
\begin{equation}
 X_{\sigma}=\sigma^i(q,\gamma(q))\frac{\partial}{\partial q^i}+\sigma_i(q,\gamma(q))\frac{\partial}{\partial p_i},\qquad X_{\sigma}^{\gamma}=\sigma^i(q,\gamma(q))\frac{\partial}{\partial q^i},
\end{equation}
respectively. 
Using a one-form section $\gamma$ on $Q$, the tangent lift of the projected
vector field $X_{\sigma}^{\gamma}$ is
\begin{equation}
T\gamma \left(X_{\sigma}^{\gamma}\right)=\sigma^i\left(\frac{\partial}{\partial q^i}+\frac{\partial \gamma_j}{\partial q^i}\frac{\partial}{\partial p_j}\right)
\end{equation}
Using \eqref{hjsect}, we find an expression relating the section $\sigma$ and the vector fields as follows.
\begin{equation}\label{hjesect}
 \sigma^i(q,\gamma(q))\frac{\partial \gamma_j}{\partial q^i}(q)=\sigma_j(q,\gamma(q)).
\end{equation}
We are ready now to state the following lemma.
\begin{lemma} \label{lemma}
 Given the conditions above, we say that: the two vector fields $X_{\sigma}$ and $X_{\sigma}^{\gamma}$ are $\gamma$-related if and only if $\eqref{hjesect}$ is fulfilled.
\end{lemma}

Again, this construction can be mimicked for nonintegrable IDE that are submanifolds $E$ of a higher-order bundle, after performing the integrability algorithm given in 4.2.

\subsection{A HJ theory for IHS}

As we have summarized in subsection 
(\ref{LSTT*Q}), for every Lagrangian submanifold $E$ of $TT^*Q$, there exists a Morse family $F:TT^{*}Q\rightarrow \mathbb{R}$ generating $E$. This enables us to write $E$ locally as 
\begin{equation} \label{EandF}
E=\left \{\left(q^i,p_i;\frac{\partial F}{\partial p_i},-\frac{\partial F}{\partial q^i}\right )\in TT^*Q:\frac{\partial F}{\partial \lambda^a}=0\right\}
\end{equation}
where $F=F(q,p,\lambda)$. We cite two important studies \cite{BeTu80,MaMoMu90} related with the problem addressed in this subsection.

We introduce a differential one-form $\gamma$ on the base manifold $Q$. 
See that, $\Ima(\gamma)$ is a Lagrangian submanifold of $ T^*Q$ so that there is an inclusion $\imath:\Ima(\gamma)\mapsto T^*Q$. Use the inclusion $\imath$ in order to pull the bundle $(P,\pi,T^*Q)$ back over $\Ima(\gamma)$. By this, one arrives at a fiber bundle $(\imath^*P,\imath^*\pi,\Ima(\gamma))$. For the present case, the commutative diagram for a generic pullback bundle exhibited in (\ref{pbb}) takes the following particular form.
\begin{equation}\label{pbb-F}
  \xymatrix{
\imath^*P \ar[rr] ^{\varepsilon} \ar[dd]^{\imath^*\pi}&& P \ar[dd]^{\pi}\\ \\
\Ima(\gamma) \ar [rr]_{\imath}&& T^*Q
}
\end{equation}
Here, the total space $$\imath^*P=\left \{(\gamma(q),z)\in \Ima(\gamma)\times P : \pi(z)\in \Ima(\gamma)\right \}$$ of the pull-back bundle is a submanifold of $P$ with $\varepsilon$ is the corresponding inclusion. A local coordinate system on $\imath^*P$ can be taken as $(q,\gamma(q),\lambda)$. Although restriction of the Morse family on $\imath^*P$ should formally be written as $F\circ \epsilon$, we will abuse notation by still denoting it by $F$ ,but to highlight the difference we shall write the arguments of the function as $F=F(q,\gamma(q),\lambda)$. The submanifold generated by $F=F(q,\gamma(q),\lambda)$ is given by
\begin{equation} \label{EandF-gamma}
E\vert_{\Ima({\gamma})}=\left \{\left(q^i,\gamma_i(q);\frac{\partial F}{\partial p_i},-\frac{\partial F}{\partial q^i}\right )\in TT^*Q:\frac{\partial F}{\partial \lambda^a}=0\right\}.
\end{equation}
 Note that, if the Lagrangian submanifold $E$ is the image of a Hamiltonian vector field $X_H$, then $E\vert_{\Ima({\gamma})}$ reduces to the image space of the composition $X_H \circ \gamma$.

The submanifold $E\vert_{\Ima({\gamma})}$ exhibited in (\ref{EandF-gamma}) does not depend on the momentum variables. This enables us to project it to a submanifold $E^{\gamma}$ of $TQ$ by the tangent mapping $T\pi_Q$ as follows
\begin{equation} \label{E-gamma}
E^{\gamma}=T\pi_Q\circ E\vert_{\Ima({\gamma})}=\left \{\left(q^i,\frac{\partial F}{\partial p_i}(q,\gamma(q),\lambda)\right )\in TQ:\frac{\partial F}{\partial \lambda^a}=0\right\}.
\end{equation}
Note that the submanifold $E^{\gamma}$ defines an implicit differential equation on $Q$. We state the generalization of the Hamilton-Jacobi theorem (\ref{HJT}) as follows. 
\begin{lemma} \label{nHJT}
The following conditions are equivalent for a closed one-form $\gamma$: 
\begin{enumerate}
\item The Lagrangian submanifold $E$ in (\ref{EandF}) and the submanifold $E^{\gamma}$ in (\ref{E-gamma}) are $\gamma$-related, that is
 $$T\gamma(E^{\gamma})=E\vert_{\Ima({\gamma})}$$

\item And it is fulfilled that $dF(q,\gamma(q),\lambda) =0$,
where $F$ is the Morse family generating $E$.
\end{enumerate}
\end{lemma}
\begin{proof}
The one-form $\gamma=\gamma_idq^i$ is closed, that is, $\frac{\partial \gamma_{i}}{\partial q^j}=\frac{\partial \gamma_{j}}{\partial q^i}$. The first assertion in lemma \ref{nHJT} can be written locally as
\begin{equation} \label{lift}
\frac{\partial \gamma_{j}}{\partial q^i}\frac{\partial F}{\partial p_j}(q,\gamma(q),\lambda)+\frac{\partial F}{\partial q^i}(q,\gamma(q),\lambda)=0,
\end{equation}
with the conditions that $\partial F/ \partial \lambda^a=0$. We make a simple calculation to compute the exterior derivative of the Morse family as follows
\begin{eqnarray}\label{pfM}
dF(q,\gamma(q),\lambda)&=&\frac{\partial F}{\partial q^j}dq^j+\frac{\partial F}{\partial p_i}\gamma_{i,j}dq^j+\frac{\partial F}{\partial \lambda^a}d{\lambda^a}.
 \end{eqnarray}
 Note that, after the substitution of (\ref{lift}) into (\ref{pfM}) and by employing the symmetry $\frac{\partial \gamma_{i}}{\partial q^j}=\frac{\partial \gamma_{j}}{\partial q^i}$, we arrive at that the exterior derivative
 of $F$ vanishes when $p=\gamma(q)$. To prove the reverse direction, it is enough to repeat these steps in reverse order.
 \end{proof}

Assume now that the one-form $\gamma$ is exact so that $\gamma=dW(q)$ for some real valued function $W$ called the characteristic function on the base manifold $Q$. Then the second condition in lemma \ref{nHJT} gives the implicit Hamilton-Jacobi equation (IHJ equation)
 \begin{equation}\label{iHJEq}
   F\left(q,\frac{\partial W}{\partial q}, \lambda\right) =\text{cst}, \qquad \frac{\partial F}{\partial \lambda^a}\left(q,\frac{\partial W}{\partial q},\lambda\right)=0.
 \end{equation} 
In the lemma (\ref{nHJT}), if the Lagrangian submanifold $E$ is the image of a Hamiltonian vector field $X_H$, then $E^\gamma$ becomes the image space of the vector field $X_H^\gamma$ in \eqref{gammarelated} and  we retrieve the classical HJ theory given in (\ref{Xg}).

It is possible to generalize Lemma (\ref{nHJT}) by replacing the image space $\Ima{\gamma}$ by an arbitrary Lagrangian submanifold $S$ of $T^*Q$. Note that, according to the generalized Poincar\'{e} lemma, there exists a Morse family $W$ on the total space of a smooth bundle $(R,\tau,Q)$ generating the Lagrangian submanifold $S$. We equip the total space $R$ with the coordinates $(q^i,\mu^\alpha)$. Then we have that the Lagrangian submanifold $S$ can be written as 
\begin{equation}
S=\left\{ \left(q^i,\frac{\partial W} {\partial q^i} (q,\mu)\right)\in T^*Q: \frac{\partial W} {\partial \mu^\beta} (q,\mu)  =0 \right \}.
\end{equation}
Now, the inclusion in the diagram (\ref{pbb-F}) becomes $\imath:S\mapsto T^*Q$. In this case, the restriction of the Morse family $F$ generating the submanifold $E$ of $TT^*Q$ to the inclusion $\epsilon$ defines the following submanifold 
\begin{equation} \label{EandF-S}
E\vert_{S}=\left \{\left(q^i,\frac{\partial W} {\partial q^i};\frac{\partial F}{\partial p_i},-\frac{\partial F}{\partial q^i}\right )\in TT^*Q:\frac{\partial F}{\partial \lambda^a}=0,\frac{\partial W} {\partial \mu^\beta} (q,\mu)  =0 \right\},
\end{equation}
 where $W=W(q,\mu)$ and $F=F(q,\frac{\partial W} {\partial q^i}(q,\mu),\lambda)$. The submanifold $E\vert_{S}$ does not depend on the momentum variable $p$ explicitly. So that, its projection $E^S$ to the tangent bundle $TQ$ by means of $T\pi_Q$ is well-defined and given by
 \begin{equation} \label{E-S}
E^{S}=\left \{\left(q^i;\frac{\partial F}{\partial p_i}(q,\frac{\partial W}{\partial q}(q,\mu),\lambda)\right )\in TQ:\frac{\partial F}{\partial \lambda^a}=0, \frac{\partial W}{\partial \mu^\beta}=0\right\}.
\end{equation}
We are now ready to state a generalization of the lemma (\ref{nHJT}) as follows.
\begin{lemma} \label{nHJT-2}
Let $S$ be a Lagrangian submanifold of $T^*Q$ generated by Morse family $W=W(q,\mu)$ defined on the total space of a bundle $(R,\tau,Q)$. The following conditions are equivalent 
\begin{enumerate}
\item The Lagrangian submanifold $E$ in (\ref{EandF}) and the submanifold $E^{S}$ in (\ref{E-S}) are $S$-related, that is $$T(dW\vert_\mu)(E^{S})=E\vert_{S}$$ for every $\mu$, where $E\vert_{S}$ is in (\ref{EandF-S}).
\item $dF\left(q,\frac{\partial W}{\partial q}(q,\mu),\lambda\right) =0$ for all $\mu$. Here, $F$ is the Morse family generating $E$.
\end{enumerate}
\end{lemma}
We only give some clues instead of writing the whole proof of this lemma since it is very similar to the proof of lemma \ref{nHJT}. The closedness of the one-form $\gamma$ is replaced by the commutativity of the second partial derivatives of the Morse family $W$ with respect to $q$.

Let us comment on the notation $T(dW\vert_\mu)$ as well. If the fiber variable $\mu$ is frozen, then the exterior derivative $dW\vert_\mu$ of the Morse family $W$ becomes a differentiable mapping from $Q$ to $T^*Q$ and its tangent mapping  $T(dW\vert_\mu)$ goes from $TQ$ to $TT^*Q$. We remark that this last comment is a generalization of the Hamilton-Jacobi theory for the Lagrangian submanifolds studied in \cite{BaLeDi13} as well.

\begin{example}
  Let $V$ be a nonholonomic $k$-dimensional distribution on $Q$ spanned by the vector fields $X_a$. We define the following Morse family on the total space of the fiber bundle $T^*Q\times \mathbb{R}^k$
  \begin{equation}\label{}
    F(q,p,\lambda)= p_i \lambda^a X_a^i(q).
  \end{equation}
  Here $F$ is a Morse family and determines a Lagrangian submanifold of $TT^*Q$ given by
  \begin{equation}\label{}
    \dot{q}^i= \lambda^a X_a^i(q), \qquad  \dot{p}_j= -p_i \lambda^a \frac{\partial X_a^i}{\partial q^j}, \qquad p_i X_a^i(q)=0.
  \end{equation}
  This system is integrable according to \cite{Ja00}. The corresponding Hamilton-Jacobi equation is computed to be
    \begin{equation}\label{}
   \frac{\partial W}{\partial q^i} \lambda^a X_a^i(q)=\text{cst}.
  \end{equation}

\end{example}

\section{Complete solutions of the HJ equation for IHS}

Before writing the complete solution of a HJ equation of IHS, we first investigate the complete solutions of HJ equations for explicit systems in terms of Morse families and Lagrangian submanifolds. 

\subsection{Lagrangian submanifolds generated by complete solutions}
In the classical sense, a solution $W$ of the HJ equation (\ref{hje}) is called complete if it depends on additional variables that equal in number to the dimension of the base manifold $Q$ \cite{Ar89}. To illustrate this, we start by considering two copies of the configuration manifold and denote them by $Q$ and $\bar{Q}$. Endow these manifolds with local coordinates $(q^i)$ and $(\bar{q}^j)$, respectively. The number of arbitrary parameters for the general solution is given by $j$, which does not necessarily equal $i$. A complete solution is a real valued function $W=W(\bar{q},q)$ on the product space $\bar{Q}\times Q$ that resolves the HJ equation (\ref{hje}). This function generates three different Lagrangian submanifolds, let us show them.

Construct the fiber bundle structure $(\bar{Q}\times Q,\rho,Q)$. Here, the bundle projection $\rho$ is assumed to be  a projection to the second factor. As discussed previously, a real valued function $W=(\bar{q},q)$ on the total space $\bar{Q}\times Q$ is called a Morse family if the matrix $[\partial ^{2}W/\partial \bar{q}^i \partial q^j]$ is non-degenerate. We draw the following diagram
 \begin{equation} \label{W-MF-1}
  \xymatrix{
\mathbb{R}&&\bar{Q}\times Q \ar[dd]^{\rho}\ar[ll]^{W}&& T^*Q \ar[dd]^{\pi_Q}\\ \\ &&
Q  \ar@{=}[rr]&& Q
}
\end{equation}
 In this case, the Morse family $W$ defines a Lagrangian submanifold of $T^*Q$ given by
  \begin{equation}\label{cs}
  S=\left\{\left((q^i,\frac{\partial W}{\partial q^i}(\bar{q},q)\right)\in T^*Q: \frac{\partial W}{\partial \bar{q}^i}(\bar{q},q)=0 \right\}.
\end{equation}

Note that, by changing the roles of $Q$ and $\bar{Q}$, we may define a bundle structure $(\bar{Q}\times Q,\bar{\rho},\bar{Q})$ over the manifold $\bar{Q}$ and obtain a diagram similar to (\ref{W-MF-1}). In this case, $W$ defines a Lagrangian submanifold $\bar{S}$ of $T^*\bar{Q}$ as follows
  \begin{equation}\label{cs-bar}
  \bar{S}=\left\{\left(\bar{q}^i,\frac{\partial W}{\partial \bar{q}^i}(\bar{q},q)\right)\in T^*\bar{Q}: \frac{\partial W}{\partial {q}^j}(\bar{q},q)=0 \right\}.
  \end{equation}

Another Lagrangian submanifold generated by $W$ is the result of  following observation. The cotangent bundle $T^*(\bar{Q}\times Q)=T^*\bar{Q}\times T Q$ of the product space is a symplectic manifold equipped with the symplectic two-form $\omega_{\bar{Q}}\ominus \omega_Q$ \cite{We77}. It is evident that image of the exterior derivative $dW$ of a complete solution $W$ is a Lagrangian submanifold 
\begin{equation}\label{ImW}
\hat{S}=\left\{ \left(\bar{q},q;\frac{\partial W}{\partial \bar{q}},\frac{\partial W}{\partial q}\right)\in T^*(\bar{Q}\times Q) \right \}.
\end{equation}
 It is known that a Lagrangian submanifold of $T^*(\bar{Q}\times Q)$ defines a symplectomorphism between  $T^*\bar{Q}$ and $ T^* Q$. The Morse family $W=W(\bar{q},q)$ generates a symplectomorphism according to the following identity
\begin{equation}
\theta_{\bar{Q}}\ominus \theta_Q  =\bar{p}_{i}d\bar{q}^{i}-p_{i}dq^{i}=d
W\left( q,\bar{q}\right),   \label{tee-}
\end{equation}
where we assume the Darboux' coordinates on the cotangent bundles. In the local picture, the induced symplectomorphism is given by
\begin{equation} \label{cd}
\varphi:T^{\ast }\bar{Q} \rightarrow T^{\ast }Q:\left( \bar{q}^{i},\frac{\partial W}{\partial
\bar{q}^{j}}\right)\rightarrow \left( q^{i},-\frac{\partial W}{%
\partial q^{j}}\right).
\end{equation}%

\subsection{Complete solutions for the case of implicit Hamiltonian systems}

Let us first try to geometrize the complete solutions of the HJ equation for explicit systems. A function $W=W(\bar{q},q)$ is a complete solution of the HJ equation (\ref{hje}) if the Hamiltonian function $H$ is constant when it is restricted to $S$ exhibited in (\ref{cs}). That is, a complete solution $W$ is the one satisfying 
\begin{equation}
H\left(q,\frac{\partial W}{\partial q}(\bar{q},q)\right)=\text{cst}, \quad \frac{\partial W}{\partial \bar{q}}=0.
\end{equation}
Using the symplectic diffeomorphism $\varphi$ in (\ref{cd}) generated by the function $W$, we pull the function $H$ back to $T^*\bar{Q}$ and see that $\varphi^*H$ is a constant. So that the dynamics generated by $\varphi^*H$ is trivial. 

Now, assume that, we have a Lagrangian submanifold $E$ of $TT^*Q$. Then there exists a Morse family $F$ generating $E$. A complete solution of the implicit Hamilton-Jacobi equation (\ref{iHJEq}) is a smooth function $W$ satisfying 
\begin{equation}
F\left(q,\frac{\partial W}{\partial q},\lambda\right)=cst, \quad \frac{\partial W}{\partial \bar{q}} (\bar{q},q) =0, \qquad \frac{\partial F}{\partial \lambda} \left(q,\frac{\partial W}{\partial q},\lambda\right)=0.
\end{equation}
We aim to pull the dynamics $E$ or the Morse family $F$ back to $T^*\bar{Q}$. To achieve this goal, we recall the definition of the pullback bundle in (\ref{pbb}) and apply it to the diffeomorphism (\ref{cd}). This way we obtain a fiber bundle structure $(\varphi^*P,\varphi^*\pi,T^{\ast }\bar{Q})$ where the total space is defined to be $$\varphi^*P=\{(\bar{z},r)\in T^{\ast }\bar{Q}\times P :\varphi(z)=\pi(r)\}$$
equipped with the induced coordinates $(\bar{q}^i,\bar{p}^i,\lambda)$, and $\varphi^*\pi$ is the projection to the second factor. We draw the following commutative diagram in order to summarize the discussion.
\begin{equation}\label{pbb-GS}
  \xymatrix{
\varphi^*P \ar[rr] ^{\hat{\varphi}} \ar[dd]^{\varphi^*\pi}&& P \ar[dd]^{\pi}\\ \\
T^{\ast }\bar{Q} \ar [rr]_{\varphi}&& T^*Q
}
\end{equation}
Here, $\hat{\varphi}$ is a diffeomorphism and in the local coordinates reads
\begin{equation} \label{cd--}
\hat{\varphi}:\varphi^*P  \leftrightarrow P :\left( \bar{q}^{i},\frac{\partial W}{\partial
\bar{q}^{i}},\lambda \right) \leftrightarrow \left( q^{i},-\frac{\partial W}{%
\partial q^{i}},\lambda \right) .
\end{equation}%
The pullback $\bar{F}=F \circ \hat{\varphi} $ of the Morse function $F$ by $\hat{\varphi}$ is a Morse family on the total space of the pullback bundle $(\varphi^*P,\varphi^*\pi,T^*\bar{Q})$. Note that $\bar{F}$ is a constant function and the implicit differential equation generated by $\bar{F}$ renders trivial dynamics. 

Generalizing, the most general form of a Lagrangian submanifold of $T^*\bar{Q}\times T^*Q$ is generated by a Morse family $W$ defined on the total space of the fiber bundle $(R,\tau,M)$ where the base manifold $M$ is a submanifold of $\bar{Q} \times Q$. Let us depict it in a diagram. 
\begin{equation}
\xymatrix{
\mathbb{R}&& R \ar[dd]^{\tau}\ar[ll]_{W}&& T^*\bar{Q}\times T^*Q \ar[dd]^{\pi_{(\bar{Q}\times Q)}}
\\
\\
&& M \ar [rr]_{\imath}&& \bar{Q} \times Q
}
\end{equation}
A complete solution to the implicit Hamilton Jacobi equation (\ref{iHJEq}) is a Morse family $W$ defined on the total space $R$. Note that, the Morse family $F$ generating the dynamics on $E$ reduces to a constant function on the Lagrangian submanifold generated by $W$.

Let us now depict the situation in coordinates. Assume that a submanifold $M$ of $\bar{Q} \times Q$ is defined by a number ``$l$" of equations
\begin{equation}
U^{a}\left( \bar{q},q\right) =0,\qquad a=1,...,l.
\end{equation}%
Consider a real function $W'$ on the submanifold $M$ and define an arbitrary continuation $$W=W'+\nu_a U^a$$ of $W'$ to the product space $P=\bar{Q} \times Q\times \Lambda$ where $(\nu_a)$'s are the Lagrange multipliers defining a local coordinate system for $\Lambda$. This $W$ is a complete solution of the implicit Hamiltonian dynamics (\ref{nHJT}) generated by $F$ if
$$ F\left(q^i,\frac{\partial W'}{\partial q^{i}},\lambda\right)=cst, \qquad U^{a}\left( \bar{q}, q\right)=0,\qquad a=1,..,l. $$
An implicit description of a Lagrangian submanifold of $T^*\bar{Q}\times T^*Q $ generated by $W$
or equivalently of the corresponding diffeomorphism $\varphi $ can be computed by
\begin{equation}
\theta_Q \ominus \theta_{\bar{Q}} =p_{i}dq^{i}-\bar{p}_{i}d\bar{q}^{i}=d\left(
W'\left( q,\bar{q}\right) +\nu _{a}U^{a}\left( q,\bar{q}\right)
\right),   \label{tee}
\end{equation}%
where the $\theta_{\bar{Q}}$ and $\theta_Q$ are the canonical one-forms on $\bar{Q}$ and $Q$, respectively. In this case, the momenta $p\in T_q^{*}Q$ and $\bar{p}\in T_{\bar{q}}^{*} \bar{Q}$ can be explicitly stated as \begin{eqnarray}
\bar{p}_{i} &=&-\frac{\partial W'}{\partial \bar{q}^{i}}-\nu _{a}%
\frac{\partial U^{a}}{\partial \bar{q}^{i}}  \notag \\
p_{i} &=&\frac{\partial W'}{\partial q^{i}}+\nu _{a}\frac{%
\partial U^{a}}{\partial q^{i}}  \notag \\
U^{a}\left( \bar{q}, q\right)  &=&0,\qquad a=1,..,l.
\end{eqnarray}%

\section{Application to degenerate Lagrangian systems}

\subsection{Lagrangian dynamics}

A Lagrangian function $L$ is a real valued function on $TQ$. Consider the coordinates $(q^i,\dot{q}^i)$ on $TQ$ induced those from $Q$. We define a vertical endomorphism $S$ given by $S=\partial/\partial \dot{q}^i\otimes dq^i$. Note that, $S$ is a $(1,1)$-tensor field on $Q$. In terms of $S$ the Cartan one-form $\theta_L$ is defined to be $$\theta_L=S^{*}(dL)=\left(\partial L / \partial {\dot q}^i\right)dq^i$$
where $dL$ is the exterior derivative of a Lagrangian density.
The Cartan two-form derivates from the Cartan one-form $\omega_L=-d\theta_L$. See that $\omega_L$ is symplectic if
the Hessian matrix
 \begin{equation}\label{hess}
  \left(W_{ij}\right)=\left(\frac{\partial^2 L}{\partial \dot{q}^i \partial \dot{q}^j}\right)
 \end{equation}
is not singular. In this case, the fiber derivative (or the Legendre transformation)
 \begin{equation}\label{FD}
  \mathbb{F}L:T Q\mapsto T ^*Q:(q^i,\dot{q}^j)\mapsto \left(q^i, \frac{\partial L}{\partial \dot{q}^j}\right)
\end{equation}
becomes a symplectomorphism relating the Cartan two-form $\omega_L$ on $TQ$ and the canonical symplectic two-form $\omega_Q$ on $T^{*}Q$. The Lagrangian is said to be hyperregular if the fiber derivative $\mathbb{F}L$ is a global diffeomorphism.

The energy is defined as
$E_L=\Delta(L)-L$, a real valued function on $TQ$ where the Liouville vector field is $\Delta=\dot{q}^i \partial/\partial \dot{q}^i$.
The Hamiltonian is retrieved through 
\begin{equation} \label{canHam}
H(q,p)=E_L\circ \mathbb{F}L^{-1}.
\end{equation}
If the Lagrangian is regular, or equivalently, if $\omega_L$ is symplectic, then the Lagrange equations can be expressed geometrically as
 \begin{equation} \label{Lvf}
  \iota_{\xi_L}\omega_L=dE_L,
 \end{equation}
 whose solution $\xi_L$ is called a Euler--Lagrange vector field explicitly given by
 \begin{equation}
  \xi_L=\dot{q}^{i}\frac{\partial}{\partial q^i}+\xi^i(q,\dot{q})\frac{\partial}{\partial {\dot q}^i}.
 \end{equation}
 The integral curves $(q^i(t),\dot{q}^i(t))$ are lifts of their projections $(q^i(t))$ on $Q$ and are solutions of the system of differential equations
 \begin{align}
  \frac{dq^i(t)}{dt}&=\dot{q}^i, \qquad \frac{d\dot{q}^i(t)}{dt}=\xi^i,
 \end{align}
 which is equivalent to a second-order differential equation
 \begin{equation}\label{sode}
  \frac{d^2 q^i(t)}{dt^2}=\xi^i.
 \end{equation}
 The curves $(q^i(t))$ in $Q$ are called the solutions of $\xi_L$ that correspond with the solutions of the Euler--Lagrange equation
 \begin{equation} \label{EL}
  \frac{d}{dt}\left(\frac{\partial L}{\partial \dot{q}^i}\right)=\frac{\partial L}{\partial q^i}.
 \end{equation}

If the Lagrangian is regular, then the fiber derivative (\ref{FD}) has the following geometry.
\begin{center}
 \begin{tikzcd}
  &(TTQ,\omega_L) \arrow[r, "T{\mathbb{F}L}"]\arrow[dr, "\tau_{TQ}"]  &(TT^{*}Q,\omega_Q^{T}) \arrow[r, "\alpha_Q"]  &T^{*}TQ  \\
  & &TQ\arrow[ur, "dL"] \arrow[lu,"\xi_L",bend left=20]\arrow[r, "L"] &\mathbb{R}
 \end{tikzcd}
\end{center}
In this case, the Hamiltonian vector field $X_H$ associated with the Hamiltonian function $H$ in (\ref{canHam}) and $\xi_L$ presented in (\ref{Lvf}) are related as $$ X_H \circ \mathbb{F}L =T\mathbb{F}L\circ \xi_L.$$
The diffeomorphisms $\alpha_Q$ and $\beta_Q$ maps Lagrangian submanifolds into Lagrangian submanifolds,
\begin{equation*}
 \alpha_Q(\Ima(X_H))=\Ima(dL), \qquad \beta_Q \circ \alpha^{-1}_Q(\Ima(dL))=\Ima(dH),
\end{equation*}
whereas the Hamiltonian and the Lagrangian vector fields are related to their
corresponding Lagrangian submanifolds as
\begin{equation*}
\beta_Q \circ X_H= dH, \qquad  \alpha_Q\circ T\mathbb{F}L\circ \xi_L=dL,
\end{equation*}
respectively.

\subsection{Lagrangian dynamics as a Lagrangian submanifold}

In this section we depict the geometric interpretation of a HJ theory for Lagrangian dynamics. For it, we present the EL equations in terms of Morse families and special symplectic structures.

Recall the special symplectic structure  on the left side of Tulczyjew's triple.
\begin{equation} \label{LagLag}
 \xymatrix{
& T^*TQ \ar[d]_{\pi_{TQ}}&TT^*Q \ar[l]_{\alpha_Q} \ar[dl]^{\tau_{TQ}}\\ \mathbb{R}& \ar[l]_{L}
T^*Q
}.
\end{equation}
In the induced local picture on $TT^*Q$, by following the procedure presented in subsection (\ref{LS-MF}), we compute the Lagrangian submanifold $E$ generated by the Lagrangian $L$ as
\begin{equation}\label{La}
E=\left \{\left(q^i,\frac{\partial L}{\partial \dot{q}^i};\dot{q}^i,\frac{\partial L}{\partial q^i}\right)\in TT^*Q \right \}
\end{equation}
which is equivalent to the EL equations (\ref{EL}).
We can generate this Lagrangian
submanifold from the right wing (the Hamiltonian
side) of the triple (\ref{T}) by defining a proper Morse family $%
F^{L\rightarrow H}$ on the Pontryagin bundle $PQ=TQ\times_{Q}T^{\ast}Q$ over $T^{\ast}Q$. 
On a local chart, the energy function
\begin{equation}  \label{HMF}
F( q,p,\dot{q}) =p_i\dot{q}^i- L (q,\dot{q})
\end{equation}
satisfies the requirements of being a Morse
family. Hence, $F$ generates a Lagrangian
submanifold of $T^{\ast}T^{\ast}Q
$ as defined in equation (\ref{LagSub}). In coordinates $(q^i,p_i,\alpha_i,\beta^i)$ of $T^{\ast}T^{\ast}Q$, this Lagrangian submanifold is given by

\begin{equation}\label{Morsefam}
\alpha_i=\frac{\partial F}{\partial q^i}=-\frac{\partial L}{\partial q^i},
\quad
\beta^i= \frac{\partial F}{\partial p_i}=\dot{q}^i,
\quad
\frac{\partial F}{\partial \dot{q}^i}=p_i-\frac{\partial L}{\partial \dot{q}^i}=0
\end{equation}
The inverse of the isomorphism $\beta_Q$ maps this Lagrangian submanifold to the Lagrangian submanifold $E$ presented in \eqref{La}. Here, we record the following diagram for this.
 \begin{equation} \label{Morse-Gen-La-Ha}
 \xymatrix{
\mathbb{R}& PQ \ar[d]^{\pi}\ar[l]^{F}& T^*T^*Q \ar[d]_{\pi_{T^*Q}}&TT^*Q \ar[l]_{\beta_Q} \ar[dl]^{\tau_{T^*Q}}\\ &
T^*Q \ar@{=}[r]& T^*Q
}
\end{equation}

 For regular cases, the Morse family $F$ on $PQ
$ can be reduced to a Hamiltonian function $H$ on $T^{\ast}Q$. For
degenerate cases, a reduction of the total space $PQ$ to a subbundle larger
than $T^{\ast}Q$ is possible depending on degeneracy level of Lagrangian
function \cite{Be11}. There exists an intrisecally geometric procedure for dealing with constraints in Hamiltonian and Lagrangian mechanics. It has been available since 1979, with advantages over the Dirac-Bergman algorithm, it is
the Gotay-Bergman algorithm \cite{gotay2,gotay5,gotay3,gotay0} (read Appendix for brief description of the method).

\subsection{Hamilton-Jacobi theory for degenerate Lagrangians}

To write the associated HJ equation of the EL equations generated by (possibly) degenerate Lagrangian densities, we apply lemma \ref{nHJT} to the Lagrangian submanifold presented in (\ref{La}). Accordingly, we arrive at that the implicit HJ equation
$$F(q,\gamma(q),\dot{q})=\gamma_i(q)\dot{q}^i- L (q,\dot{q})=\text{cst}$$
for a closed one-form $\gamma=\gamma_i(q) dq^i$. 
Taking the exterior derivative of this equation, we arrive at the following local picture of the Hamilton-Jacobi equation for a Lagrangian $L$ 
 \begin{equation} \label{DefnG}
   \dot{q}^i\frac{\partial \gamma_i}{\partial q^j}(q)-\frac{\partial L}{\partial q^j}(q,\dot{q})=0,\qquad \gamma_i (q)-\frac{\partial L}{\partial \dot{q}^i}(q,\dot{q})=0.
 \end{equation}
To illustrate this, we propose two particular problems \cite{deLeMaDeDiVa13}.
\begin{example}\normalfont

Consider the degenerate Lagrangian $L$ on $T\mathbb{R}^3 $ given by
$$ L(q,\dot{q})=L(q^1,q^2,q^3,\dot{q}^1,\dot{q}^2,\dot{q}^3)=\frac{1}{2}(\dot{q}^1+\dot{q}^2)^2, $$ and the Whitney bundle $T\mathbb{R}^3\oplus T^*\mathbb{R}^3$ fibered on $T^*\mathbb{R}^3$
and parameterized by $(q^1,q^2,q^3,\dot{q}^1,\dot{q}^2,\dot{q}^3,p_1,p_2,p_3)$. The Lagrange multipliers correspond with $(\lambda^i)= (\dot{q}^1,\dot{q}^2,\dot{q}^3)$.
Following (\ref{HMF}), we define the Morse family
$$ F(q,\dot{q},p)=\dot{q}^1p_1+\dot{q}^2p_2+\dot{q}^3p_3-\frac{1}{2}(\dot{q}^1+\dot{q}^2)^2.$$
that generates a Lagrangian submanifold $E$ of $TT^*\mathbb{R}^3$. Explicitly,
\begin{eqnarray}\label{SEx1}
 E=\{ (q^1,q^2,q^3,p_1,p_2,p_3;\dot{q}^1,\dot{q}^2,\dot{q}^3,0,0,0)\in TT^*\mathbb{R}^3 \notag \\ :p_1=\dot{q}^1+\dot{q}^2, p_2=\dot{q}^1+\dot{q}^2, p_3=0 \}.
\end{eqnarray}
It is evident that the projection of $E$ to $T^*\mathbb{R}^3$ results with the following $4$ dimensional submanifold $$C=\{(q^1,q^2,q^3,p_1,p_2,p_3)\in T^*\mathbb{R}^3:p_1=p_2,p_3=0\}.$$

Consider now the closed one-form $\gamma:\mathbb{R}^3\mapsto T^*\mathbb{R}^3$. According to the Lagrangian HJ theorem (\ref{nHJT}),
the system (\ref{DefnG}) in this particular case takes the form:
\begin{equation}\label{sysEx1}
\begin{cases}
 [1]\quad  \dot{q}^1\frac{\partial \gamma_1}{\partial q^1}+\dot{q}^2\frac{\partial \gamma_2}{\partial q^1}+\dot{q}^3\frac{\partial \gamma_3}{\partial q^1}=0, \\
 [2]\quad  \dot{q}^1\frac{\partial \gamma_1}{\partial q^2}+\dot{q}^2\frac{\partial \gamma_2}{\partial q^2}+\dot{q}^3\frac{\partial \gamma_3}{\partial q^2}=0, \\
  [3]\quad \dot{q}^1\frac{\partial \gamma_1}{\partial q^3}+\dot{q}^2\frac{\partial \gamma_2}{\partial q^3}+\dot{q}^3\frac{\partial \gamma_3}{\partial q^3}=0, \\
  [4]\quad \gamma_1-\dot{q}^1-\dot{q}^2=0 \\
  [5]\quad \gamma_2-\dot{q}^1-\dot{q}^2=0 \\
  [6]\quad \gamma_3=0.
  \end{cases}
\end{equation}
It is immediate to see from equations [4] and [5] that  $\gamma_1=\gamma_2$. If $\gamma$ is closed and $\gamma_3=0$, one obtains that $\gamma_1$ and $\gamma_2$ are independent of $q^3$.
Then [3] in (\ref{sysEx1}) is automatically satisfied. The substitution of [4] into [1] and [2] results in
$$\gamma_1\frac{\partial \gamma_1}{\partial q^1}=0,\qquad \gamma_1\frac{\partial \gamma_1}{\partial q^2}=0.$$
A nontrivial solution is possible if $\gamma_1$ and $\gamma_2$ are constants. Hence, we record the one-form
$$\gamma(q)=(q^1,q^2,q^3,c,c,0).$$

In terms of submanifolds, that is, according to the first condition in theorem \ref{nHJT}, the picture is the following. 
The constant character of the one-form defines a constraint $\dot{q}^1+\dot{q}^2=c$, on the velocity variables. We first restrict the submanifold $E$ to the image space of the $\gamma$, we arrive at 
\begin{equation}
E\vert_{\Ima(\gamma)}=\left \{ (q^1,q^2,q^3,c,c,0;\dot{q}^1,\dot{q}^2,\dot{q}^3,0,0,0)\in TT^*\mathbb{R}^3 \notag \\ :c=\dot{q}^1+\dot{q}^2 \right \}
\end{equation}
whose generic version is given in (\ref{EandF-gamma}). The projection of $E\vert_{\Ima(\gamma)}$ to the tangent bundle $T\mathbb{R}^3$ by $T\pi_{\mathbb{R}^3}$ results in a five dimensional submanifold
$$E^\gamma=\left\{ (q^1,q^2,q^3,\dot{q}^1,\dot{q}^2,\dot{q}^3)\in T\mathbb{R}^3 \quad \text{with} \quad \dot{q}^1+\dot{q}^2=c\right \},$$
of the tangent bundle $T\mathbb{R}^3$. The system of equations (\ref{sysEx1}) is equivalent to saying that the tangent lift of $E^\gamma$ by the tangent mapping $T\gamma$ equals  $E\vert_{\Ima(\gamma)}$. It is indeed immediate to see that $T\gamma (E^\gamma)=E\vert_{\Ima(\gamma)}$.

In terms of vector fields, the situation is as follows. Consider a section $\sigma$ of the tangent bundle $\tau_{T^*Q}$ given by $$\sigma(q^1,q^2,q^3,p_1,p_2,p_3)=(q^1,q^2,q^3,p_1,p_2,p_3;c-\dot{q}^2,\dot{q}^2,\dot{q}^3,0,0,0) $$
on the intersection $C\cap\Ima(\gamma)$. Note that, $\Ima(\sigma)\subset E$. Accordingly, we write the following vector field 
$$X_\sigma = c \frac{\partial}{\partial q^1}+\dot{q}^3\circ \gamma(q) \frac{\partial}{\partial q^3}+\dot{q}^2\circ \gamma(q) \left(\frac{\partial}{\partial q^2}-\frac{\partial}{\partial q^1}\right).$$
We project this vector field by $T\pi_Q$ and arrive at the vector field $X^\gamma_\sigma$, which is locally the same as $X_{\sigma}$. Composing with the section $\gamma$,
$$ T\gamma(q^1,q^2,q^3,\dot{q}^1,\dot{q}^2,\dot{q}^3)=(q^1,q^2,q^3;c,c,0,\dot{q}^1,\dot{q}^2,\dot{q}^3,0,0,0)$$
This shows that
\begin{equation}
T\gamma(X_{\sigma}^{\gamma})=X_{\sigma}
\end{equation}  
is obviously fulfilled.
\end{example}

\begin{example}\normalfont

Consider the Lagrangian $L$ on $T\mathbb{R}^2 $ given by
$$ L(q,\dot{q})=L(q^1,q^2,\dot{q}^1,\dot{q}^2)=\frac{1}{2}(\dot{q}^1)^2+q^2(q^1)^2. $$
To recast the EL system generated by this Lagrangian density, we simply define the following Morse family on the Whitney sum $T\mathbb{R}^2\oplus T^*\mathbb{R}^2$
\begin{equation} \label{Morse-2}
F(q,p,\dot{q})=p_1\dot{q}^1+p_2\dot{q}^2-\frac{1}{2}(\dot{q}^1)^2-q^2(q^1)^2. 
\end{equation}
This family generates the following Lagrangian submanifold
\begin{equation} \label{E-2}
E=\left \{(q^1,q^2,p_1,p_2;\dot{q}^1,\dot{q}^2,2q^2q^1,(q^1)^2) \in TT^*\mathbb{R}^2:p_1=\dot{q}^1, p_2=0 \right \}.
\end{equation}
The projection of $E$ onto the cotangent bundle $T^*\mathbb{R}^2$ by the tangent bundle projection $\tau_{T^*Q}$ results with the following submanifold
\begin{equation}\label{C3}
  C=\{(q^1,q^2,p_1,p_2)\in T^*\mathbb{R}^2:p_2=0\}.
\end{equation}

According to theorem \ref{nHJT}, we now introduce a closed one-form $\gamma$ on $\mathbb{R}^2$ and require that the Morse family $F$ in (\ref{Morse-2}) is constant on the image space, that is 
$$F(q,\gamma(q),\dot{q})=
\gamma_1(q)\dot{q}^1+\gamma_2(q)\dot{q}^2-\frac{1}{2}(\dot{q}^1)^2-q^2(q^1)^2=cst.$$
For this case, the Hamilton Jacobi equations (\ref{DefnG}) turn out to be the following set
\begin{equation}\label{DefnEx3}
\begin{cases}
 [1]\quad \dot{q}^1\frac{\partial \gamma_1(q)}{\partial q^1}+\dot{q}^2\frac{\partial \gamma_2(q)}{\partial q^1}-2q^2q^1=0 \\
 [2]\quad \dot{q}^1\frac{\partial \gamma_1(q)}{\partial q^2}+\dot{q}^2\frac{\partial \gamma_2(q)}{\partial q^2}-(q^1)^2=0 \\
 [3]\quad \gamma_1(q)-\dot{q}^1=0 \\
 [4]\quad \gamma_2(q)=0.
 \end{cases}
\end{equation}
The closedness of $\gamma$, together with equation [4], imply that $\gamma_1$ depends only on $q^1$. If we substitute this and equations [3] and [4] in [1] and [2], then we arrive at a partial differential equation and a constraint
\begin{eqnarray}\label{}
\gamma_1(q)\frac{\partial \gamma_1(q)}{\partial q^1}-2q^2q^1=0, \qquad q^1=0 \notag.
\end{eqnarray}
Note that the constraint $q^1=0$ implies that the system is automatically satisfied for any function $\gamma_1=\gamma_1(q^1)$. 
Then, the one-form is described by
\begin{equation} \label{gamma-2}
\gamma(q)=(q^1,q^2;q^1,0).
\end{equation}

Let us assume that $\gamma_1(q^1)= q^1$
as it is done in \cite{deLeMaDeDiVa13}, and rewrite system (\ref{DefnEx3}) in terms of submanifolds. If the Lagrangian submanifold $E$ in (\ref{E-2}) is restricted to the image space of $\gamma$ in (\ref{gamma-2}), then the result becomes 
 \begin{equation}
E\vert_{\Ima(\gamma)}=\left \{(q^1,q^2,q^1,0;q^1,\dot{q}^2,2q^2q^1,(q^1)^2) \in TT^*\mathbb{R}^2: {q}^1=0\right \}.
\end{equation}
This is projected to the tangent bundle $T\mathbb{R}^2$ by $T\pi_{\mathbb{R}^2}$ in order to get the reduced dynamics 
$$E^\gamma=\{(q^1,q^2;\dot{q}^1,\dot{q}^2):q^1=0,\dot{q}^1=q^1=0\}.$$
The $\gamma$-relatedness of $E$ and $E^\gamma$ can be checked with the following lift 
$$T\gamma \circ E^\gamma=(q^1=0,q^2,q^1,0;\dot{q}^1=0,\dot{q}^2,\dot{q}^1=0,0).$$
As a result, we have three constraints $q^1=0,p_1=0$ and $p_2=0$. So, the projected submanifold $C$ in \eqref{C3}
must be rectified as
$$ C=\{(q^1,q^2,p_1,p_2):q^1=0,p_1=0,p_2=0\}\subset T^*\mathbb{R}^2.$$
Here, $p_2=0$ is called the primary constraint, since it roots in the functional structure of the Lagrangian function, and the other two constraints $q^1=0$ and $p_1=0$ are called the secondary constraints.

Consider a section $\sigma$ of the tangent bundle fibration $\tau_{T^*\mathbb{R}^2}$ given by
$$ \sigma(q^1,q^2;\dot{q}^1,\dot{q}^2)=(q^1,q^2,\dot{q}^1,\dot{q}^2;\dot{q}^1,\dot{q}^2,2q^2q^1,(q^1)^2).$$
By restricting this section to the intersection $\Ima (\gamma)\cap C$, we arrive at the following vector field $$ X_\sigma=\dot{q}^2\circ \gamma(q) \frac{\partial}{\partial q^2}.$$
The projection of this vector field to $T\mathbb{R}^2$ is the vector field  $X_\sigma^\gamma$, and it looks exactly as $X_\sigma$, at least locally.  
$X_\sigma^\gamma$ and $X_\sigma$ are $\gamma$ related since the tangent lift by $\gamma$ is given by
$$T\gamma(q,\dot{q})=(q^1,q^2,q^1,0;\dot{q}^1,\dot{q}^2,\dot{q}^1,0)$$
and maps the vector field $X_\sigma^\gamma$ into $X_\sigma$.

\end{example}

\section*{Conclusions and future work}

In this work we have presented a geometric Hamilton-Jacobi theory for systems of implicit differential equations. In the general context, due to the implicit character of the equations, the lack of a vector field has been solved by the introduction of a local section $\sigma$. In the particular case of the implicit Hamiltonian dynamics, Morse families and special symplectic structures have been employed 
to derive a Hamilton--Jacobi theory in which the Morse function plays the role of the Hamiltonian.
This result has been particulary applied to singular Lagrangians. We expect further applications of the theory in constraint Hamiltonian systems, Dirac systems, and etc. 

The obtainance of a Hamilton--Jacobi theorem through reduction
 is here sketched in terms of coisotropic reduction \cite{LeMaVa17}. As  future work, we aim at stating the problem of reduction of the implicit Hamilton-Jacobi theory under the Lie group symmetry of the implicit system of differential equations.  

\section*{Acknowledgements}
This work has been partially supported by MINECO MTM 2013-42-870-P and
the ICMAT Severo Ochoa project SEV-2011-0087. One of us (OE) is grateful for ICMAT for warm hospitality where some parts of the works has been done. (OE) is also grateful Prof. Hasan G\"{u}mral for valuable discussions in the theory of Tulczyjew triples. 

\newpage

\section*{Appendix}

\subsection*{The Gotay--Nester algorithm}

The Gotay-Nester algorithm is a suitable tool for reducing the dynamics of  singular Hamiltonian or Lagrangian systems to a reduced manifold where the motion is well-defined. This algorithm was created as a generalization of the well-known Dirac-Bergman alrorithm which has local nature and does not cope with all the singularities appearing in dynamics. Let us briefly describe the Gotay-Nester algorithm \cite{gotay2,gotay5,gotay3,gotay0}.

Let us recall that given a singular Lagrangian $L:TQ\rightarrow \mathbb{R}$, the Legendre transform $FL:TQ\rightarrow T^{*}Q$
and the energy $E_L:TQ\rightarrow \mathbb{R}$, one can define a presymplectic system $(M_1=FL(TQ),\omega_1)$, 
where $\omega_1$ is the restriction of the canonical symplectic form on $T^*Q$ to $M_1$.
We will assume that $L$ is almost regular (the Legendre transformation is a submersion and surjective), then $M_1= FL(TQ)$ is a submanifold of $T^*Q$

The restriction of the Legendre mapping $FL_1 : TQ \longrightarrow M_1$  to this submanifold is a submersion with connected fibers. In this case, $M_1$ is called the submanifold of primary constraints. If $L$ is almost regular, $\ker( TFL)=\ker( \omega_L)\cap V(TQ)$, where
$V(TQ)$ denotes the vertical bundle, and the fibers are connected, a direct computation shows that $E_L$ projects onto a function $h_1:M_1 \rightarrow \mathbb{R}.$ The inclusion of this submanifold is denoted by $j_1 : M_1 \longrightarrow T^*Q$
and define $\omega_1 = j_1^* (\omega_Q)$. The dynamics in the primary constraint manifold is
\[
i_X\omega_1=dh_1,
\]
where $h_1\in C^{\infty}(M_1)$ is the projection of the energy $E_L\in C^{\infty}(TQ)$.
Now, there are two possibilities: the solution $X$ defined at all the points of $M_1$ is such that $X$ defines global dynamics and it is a solution (modulo
$\ker \, \omega_1$), in other words, there are only primary
constraints. Or, we are in need of a second submanifold $M_2$ where $\iota_X\omega_1=dh_1$
and $X\in TM_1$. But such a solution $X$ is
not necessarily tangent to $M_2$, so we have to impose an additional tangency
condition to $M_2$ and obtain a new submanifold $M_3$ along which there
exists a solution. Continuing this process, we obtain a sequence of
submanifolds
\[
\cdots M_k \hookrightarrow \cdots \hookrightarrow M_2 \hookrightarrow M_ 1
\hookrightarrow T^*Q
\]
where the general description of $M_{l+1}$ is
\[
M_{l+1}:=\{p\in M_{l} \textrm{ such that there exists } X_p\in T_pM_l \\ 
\textrm{ satisfying } i_X\omega_1=dh_1 \}.
\]
If the algorithm stabilizes at some $k$, say $M_{k+1}=M_k$, then we
say that $M_k$ is the final constraint submanifold which is denoted by $M_f$, and then there
exists a well-defined solution $X$ along $M_f$. This constraint algorithm produces a solution $X$ of the equation
$$
(i_X \, \omega_1 = dh_1)_{|M_f} \; ,
$$
where $X$ is a vector field on $M_f$. We can depict the situation in the following diagram:
\[
\xymatrix{
TQ \ar@{->}[r]^{FL} \ar[dr]_{FL_1} & T^*Q \ar[dr]^{\pi _Q} &   \\
  &  M_1 \ar@{_{(}->}[u] \ar[r]^{\pi _1} & Q \\
  &  M_2 \ar@{_{(}->}[u] \ar[r]^{\pi _2 }&\ar@{_{(}->}[u] Q_2  \\
  &      \ar@{_{(}->}[u] & \ar@{_{(}->}[u]  \\
  &      \ar@{.}[u] &\ar@{.}[u]   \\
  &  M_f \ar@{_{(}->}[u] \ar[r]^{\pi _f } & \ar@{_{(}->}[u] Q_f
      }
\]

%
%
%
%
%
%
%

\end{document}